\documentclass[final,3p,times]{elsarticle}

\usepackage{amsmath,amssymb,amsthm} 
\usepackage{graphicx,multirow} 
\usepackage{adjustbox}
\usepackage{subcaption}
\usepackage{algorithm}
\usepackage{algpseudocode}
\usepackage{makecell}
\usepackage{natbib}
\usepackage{pdflscape}
\newtheorem{proposition}{Proposition}
\usepackage{xcolor}
\usepackage{rotating}
\usepackage{multirow}
\usepackage{tabularx}
\usepackage{booktabs} 
\usepackage{longtable}
\usepackage{float} 
\newtheorem {Def}{Definition}
\usepackage{tikz}
\usepackage{hyperref}
\usepackage{titlesec}
\usepackage{caption}
\titlespacing*{\subsection}
  {0pt}   
  {7pt}   
  {4pt}   

\journal{}

\begin{document}

\begin{frontmatter}



\title{Spectral Efficiency Centrality: An Efficient Spectral Approach for Influential Node Identification in Temporal Networks}


\author[mymainaddress]{Aksa Urooj \corref{mycorrespondingauthor}}\ead{aksaurooj62@gmail.com}

\author[mymainaddress]{Iqra Altaf Gillani}
\ead{iqraaltaf@nitsri.ac.in}

\cortext[mycorrespondingauthor]{Corresponding author}

\address[mymainaddress]{Department of Information Technology, National Institute of Technology, Srinagar, Jammu and Kashmir, India}
\begin{abstract}
Centrality measures play a vital role in identifying influential nodes in evolving networks. While existing temporal centrality measures primarily rely on local structural properties or temporal paths, spectral node-removal approaches have been largely limited to static networks. To bridge this gap, we propose Spectral Efficiency Centrality (SEC), a temporal spectral centrality framework that quantifies node importance by evaluating the change in spectral radius caused by node removal across temporal snapshots. By capturing the global structural influence of nodes throughout network evolution, SEC identifies nodes that are critical for preserving the structural connectivity and efficiency of temporal networks.\par
To improve computational scalability, we further develop an efficient approximation, ASEC, based on Perron–\\
Frobenius theory and first-order eigenvalue perturbation. ASEC requires only the leading eigenpair and avoids repeated eigendecomposition, making it suitable for large temporal networks. Extensive experiments on multiple real-world temporal datasets demonstrate that SEC and ASEC outperform existing baseline centrality measures in identifying influential nodes under SI, SIS, and IC diffusion models. Statistical significance and robustness analyses further confirm their effectiveness, while ASEC offers a computationally efficient solution for large-scale temporal networks.
\end{abstract}





\end{frontmatter}

\section{Introduction}
\label{sec:Intro}
Complex systems can be effectively represented through network-based frameworks \cite{fortunato2010community} composed of interacting entities \cite{newman2003structure}\cite{albert2002statistical}\cite{boccaletti2006complex}. In such representations, nodes correspond to system components and edges capture interactions or relationships between them. Network-based modeling has been widely adopted across diverse domains, including social systems, communication networks, biological processes, transportation infrastructures, and information diffusion. By abstracting interactions into graph structures, network science enables the analysis of connectivity patterns, structural organization, and collective behavior at a global scale.\par
A fundamental task in network analysis is to determine the most significant nodes, commonly addressed through centrality \cite{brede2012networks} measures. These measures provide insights into the relative significance of individual nodes and determine their role in various applications \cite{lu2016vital}. Such applications extend across many domains, including assessing the impact of scientific journals [7], ranking sports teams or athletes \cite{callaghan2007random}\cite{chartier2011sensitivity}\cite{saavedra2010mutually}, identifying influential individuals\cite{kempe2003maximizing}\cite{ZAHOOR2026115648}, detecting critical infrastructures vulnerable to congestion \cite{guimera2005worldwide} \cite{holme2003congestion}, identifying key nodes in communities to eradicate them, like in the case of lung diseases \cite{sun2007comparative} \cite{girvan2002community} \cite{hu2010measuring} \cite{lancichinetti2010characterizing} and determining key genetic and protein targets [65], among others.\par 
Traditional centrality measures such as closeness \cite{sabidussi1966centrality}, degree  \cite{bonacich1972factoring}, betweenness \cite{freeman1978centrality}, eigenvector \cite{bonacich1987power}\cite{wasserman1994social} and, katz centrality \cite{katz1953new} have long served as benchmarks for quantifying node importance in directed, undirected, weighted\cite{opsahl2010node}, and static networks, where interactions are assumed to remain fixed over time.  However, many real-world systems, such as communication (contacts occur intermittently), social systems (relationships evolve as individuals interact at different moments), and transportation networks (routes and flows change over time), evolve with time, making static centralities insufficient for capturing their dynamic structure.\par
To overcome this, temporal networks have increasingly become a better representation of time-varying systems. In a dynamic network, interactions are organized into discrete snapshots or time-stamped events, preserving the order and duration of connections.
Several classical centrality measures originally defined for static networks have been extended to temporal settings in order to capture time-dependent structural variations\cite{grindrod2012models}. These measures 
attempt to incorporate temporal information either by aggregating snapshots with weighting schemes or by constructing multilayer representations of the network. While these extensions improve upon purely static formulations, they often rely on snapshot aggregation, or inter-layer coupling assumptions that may hide the fine-grained temporal ordering of interactions. In particular, they may not completely keep the correct order of events or accurately capture how connections happen over time in changing networks.\par

These limitations have led to the emergence of centrality measures specifically designed for temporal networks, including temporal walk centrality \cite{oettershagen2022temporal}, dynamic centrality \cite{lerman2010centrality}, coverage centrality \cite{takaguchi2016coverage}, and TempoRank\cite{rocha2014random}. Such methods explicitly account for time-respecting paths, temporal reachability, and diffusion dynamics, thereby providing a more exact representation of influence in time-evolving systems. These metrics have been widely employed in applications ranging from real online social networks \cite{daly2008social}\cite{brodka2011degree} to epidemic modeling\cite{kim2012temporal} \cite{anderson1991infectious} \cite{lloyd2009epidemic}. However, in many applications, the significance of a node is not solely determined by how frequently it interacts with others, but also by how its presence plays a significant role in sustaining the network’s global connectivity over time. Existing approaches often focus on local or path-based properties, while often overlooking the structural robustness of the network in the presence of node failures.  This temporal dimension introduces additional complexity, as structural changes at different snapshots can have non-uniform and cumulative effects on overall network behavior. The removal of a critical node can drastically reduce the efficiency of information flow by weakening the network’s largest eigenvalue (spectral radius), which directly reflects its ability to sustain diffusion processes. Thus, there is a growing need for a metric that explicitly captures a node’s contribution to the spectral robustness of temporal networks. \par
Several studies have previously explored spectral properties for assessing node importance\cite{milanese2010approximating}\cite{restrepo2006characterizing}. Wang et al. \cite{wang2011identifying} utilized the spectrum of the adjacency matrix to identify nodes critical to the community structure of static graphs, where influence was linked to variations in eigenvalues of the adjacency matrix. Similarly, Greetham et al. \cite{vukadinovic2013centrality} examined the role of spectral radius in dynamic communication networks, highlighting its correlation with network connectivity and information flow. While these works recognize the spectral radius as a global descriptor, they primarily focus on overall network-level dynamics rather than quantifying the individual contribution of each node to influence maximization across time.\par
In this work, we adopt a complementary perspective and ask a different question: how much does a node contribute to the global connectivity of a temporal network? Motivated by this, we propose Spectral Efficiency Centrality (SEC), a global centrality measure for evolving networks. SEC builds upon both Efficiency Centrality \cite{wang2017new} and spectral-radius-based node importance measures\cite{milanese2010approximating}\cite{vukadinovic2013centrality} by replacing the shortest-path-based efficiency criterion with a spectral criterion based on the normalized reduction in the network's largest eigenvalue and extending it to temporal networks through snapshot-wise evaluation. SEC evaluates a node’s role by measuring the drop in the network’s spectral radius when the node is removed, thereby quantifying each node’s contribution to overall network efficiency. Although SEC is based on the change in spectral radius after node removal, it should not be interpreted merely as a measure of spectral vitality or network robustness. Rather, it is proposed as a node centrality measure that leverages spectral properties to quantify the structural importance of individual nodes. By shifting from path-based efficiency to spectral efficiency, SEC captures global structural disruption rather than merely local communication loss, making it particularly suitable for diffusion-driven and temporally evolving systems. SEC has been shown to effectively identify nodes that contribute to global stability, particularly in disease spreading networks. \par
While SEC provides a principled measure of global node importance, its exact computation requires repeated spectral radius evaluations, which can be computationally expensive \cite{ghosh2011parameterized} for large-scale temporal networks. To address this challenge, we develop an efficient approximation of SEC (ASEC) based on Perron–Frobenius theory and first-order eigenvalue perturbation. This approximation significantly reduces computational cost while preserving the relative influence patterns captured by exact SEC, avoiding repeated eigen-decomposition and requiring only the leading eigenpair, enabling scalable analysis of real-world large-scale temporal networks.
The key contributions of this work are summarized as follows:
\begin{itemize}
    \item We propose Spectral Efficiency Centrality (SEC), a temporal centrality framework that extends spectral-radius-based node importance from static graphs to temporal networks by evaluating the normalized reduction in spectral radius across successive temporal snapshots and aggregating the resulting node importance over time.
    \item We derive an efficient Perron–Frobenius–based approximation (ASEC) using first-order eigenvalue perturbation, significantly reducing computational complexity by avoiding repeated spectral recomputation, thereby enabling scalable analysis of large temporal networks.
    \item We theoretically analyze and comprehensively evaluate SEC and ASEC through synthetic and real-world temporal networks, assessing their theoretical properties, effectiveness, robustness, scalability, and ability to identify influential nodes under multiple diffusion models (SI, SIS, and IC).
\end{itemize}

\section{Related Work}
\label{sec:relatedW}
The study of centrality measures has received considerable attention in the field of network science, providing different perspectives on node importance. Over the past decade, several centrality measures have been developed and extended to detect important nodes in temporal graphs\cite{estrada2013communicability} \cite{pan2011path}\cite{takaguchi2016coverage}, including betweenness, degree, eigenvector, and closeness centralities, as well as the k-shell decomposition method \cite{kitsak2010identification}. In this section, we critically review structural, path-based, and spectral approaches, highlighting limitations that motivate our framework. \par
\subsection{Structural and Path-Based Centralities in Temporal Networks}
Degree centrality is the easiest way to find node importance. Holme et al. \cite{holme2012temporal} provided one of the foundational discussions on how degree centrality can be adapted to dynamic settings by considering the number of time-respecting interactions rather than static connections. Uddien et al.\cite{uddin2011time} proposed a time-sensitive method for measuring degree centrality, known as Time Scale Degree Centrality (TSDC). This approach considers both the existence of links between nodes and the duration of those connections within the network. However, it needs additional processing memory to store the adjacency matrices that keeps record of the number of edges between nodes and the time at which the interaction started between them. Kostakos \cite{kostakos2009temporal} analyzed the temporal evolution of degree and interaction bursts, showing how time dependency alters connectivity patterns. Tang et al.\cite{tang2010analysing} further developed temporal network metrics, introducing time-resolved degree measures to study dynamic communication networks. Nicosia et al. \cite{nicosia2013graph} developed a degree-based method to quantify the activity and influence of nodes across temporal layers, highlighting how node importance varies with time. Kim et al. \cite{kim2012temporal} extended classical degree and other centralities into a continuous-time framework, incorporating temporal order and interaction duration. Although these methods include temporal information, they mainly measure local activity. They fail to reflect the influence of a node’s structural position on the spreading process within the network. \par
To address this, Kitsak \textit{et al.} \cite{kitsak2010identification} demonstrated that a node’s degree alone is not sufficient to determine its spreading influence. Instead, a more reliable indicator is the node’s location within the network. Li \textit{et al.} \cite{li2013efficient} developed an approach to efficiently determine and update the core values of affected nodes following graph modifications.  The algorithm iteratively removes nodes, beginning with the lowest degree, while initially assigning color 0 to all nodes. When an edge is inserted, the core number updates are determined using the standard k-core decomposition method \cite{sariyuce2013streaming}. While core-based methods include structural embedding, they are still based on connectivity structure rather than direct diffusion behavior. \par
Betweenness centrality captures how frequently a node lies on shortest paths. Recent advancements in temporal betweenness centrality have introduced various methodologies to address the challenges posed by dynamic networks.  Lee et al. introduced QUBE(Quick algorithm for Updating BEtweenness centrality) \cite{lee2012qube} to address the challenge of updating betweenness centrality in evolving graphs. Instead of recomputing betweenness for all vertices, QUBE identifies a reduced set of affected vertices based on the concept of Minimum Union Cycles (MUCs) and updates their scores efficiently. Naima \cite{naima2025temporal} distinguished between passive and active shortest paths, proposing an improved time complexity for classical betweenness centrality computations.
\begin{table}[H]
\renewcommand{\arraystretch}{6} 
\setlength{\tabcolsep}{10pt} 
\caption{Summary of work done on Temporal Networks}
\label{tab:relatedwork}
\resizebox{\textwidth}{!}{  
\begin{tabular}{ccccccc}
\toprule
\textit{Reference} & \textbf{Centrality Type} & \textbf {Objective} & \textbf {Methodology} & \textbf {Strength} & \textbf {Limitation} & \textbf {Application Area}  \\
\midrule
Holme \textit{et al.} \cite {holme2012temporal} &\multirow{4}{*}{\rotatebox[origin=c]{90}{\LARGE Degree Centrality}}  & \makecell{Introduced temporal degree\\ centrality using time-\\respecting interactions} & \makecell{Counts node interactions\\ in each snapshot} & \makecell{Tracks node activity over\\ time; simple and intuitive.} & \makecell{Cannot capture global\\ influence.} & \makecell{Social networks,\\ communication networks.}\\
Uddin \textit{et al.} \cite{uddin2011time}  & & \makecell{extends traditional degree\\ centrality by incorporating\\ both the presence and\\ duration of links\\} &\makecell{Introduced TSDC, a time-variant\\ extension of degree centrality} & \makecell{Incorporates both the presence\\ and duration of links}& \makecell{needs additional processing\\ memory to store the adjacency\\ matrices} & \makecell{Healthcare networks,\\ inter-organizational collaboration,\\ communication and social\\ networks} \\
Kostakos \textit{et al.} \cite{kostakos2009temporal} & &  \makecell{Studied temporal graphs\\ and activity bursts; analyzed\\ degree evolution.} & \makecell{Measures degree per temporal\\ layer and analyzes burst\\ patterns.} & \makecell{Highlights temporal\\ fluctuations in connectivity.} & \makecell{Ignores long-term structural\\ influence.} & \makecell{Human interaction networks,\\ social media.} \\
Tang \textit{et al.} \cite{tang2010analysing} & &  \makecell{Proposed time-resolved\\ degree measures for \\ dynamic networks} & \makecell{Computes node degree per\\ timestamp and aggregates\\ over time.} & \makecell{Captures dynamic communication\\ patterns.} & \makecell{Computationally intensive\\ for large networks.} & \makecell{Temporal communication\\ networks.} \\  \hline
Lee \textit{et al.} \cite{lee2012qube} &\multirow{4}{*}{\rotatebox[origin=c]{90}{\LARGE Betweenness Centrality}} & \makecell{Efficient update of BC} & \makecell{Update theorem + MUC sets} & \makecell{Computational speedup} & \makecell{Localized,\\ ignores efficiency} & \makecell{Social networks,\\ temporal analysis}\\
Naima \textit{et al.} \cite{naima2025temporal} & & \makecell{Introduced active and\\ passive temporal\\ betweenness centrality.} & \makecell{Analyzed passive\\ and active shortest paths\\in temporal graphs.} & \makecell{Improved time complexity\\ for classical betweenness\\ centrality.} & \makecell{Active paths computation is\\ more complex} & \makecell{Temporal networks,\\ dynamic systems.} \\
Cruciani \textit{et al.} \cite{cruciani2024mantra}&  & \makecell{Efficient approximation of\\ temporal betweenness\\ centrality} & \makecell{Sampling-based approximation\\ method  to estimate temporal\\ betweenness, reducing\\ computational complexity} & \makecell{Scalable to large\\ temporal networks;\\ maintains high accuracy} & \makecell{Limited to betweenness-type\\ measures; may underperform \\ on networks with highly\\ irregular temporal dynamics} & \makecell{Epidemic spreading analysis,\\information diffusion,} \\
Bub \textit{et al.} \cite{buss2020algorithmic} & &  \makecell{Investigated the computational\\ and algorithmic properties of\\ several variants of temporal\\ betweenness centrality.
} & \makecell{Studied strict/non-strict\\ shortest paths, and\\ fastest paths in \\temporal networks.} & \makecell{Provides insights into\\ computational complexity and\\ feasibility of different variants.} & \makecell{Some path variants are\\ P-hard, limiting practical\\ applicability in large\\ networks.} & \makecell{Temporal network analysis,\\ algorithm design,\\ network science research.}  \\ \hline 
Tang \textit{et al.} \cite{tang2010analysing} &\multirow{4}{*}{\rotatebox[origin=c]{90}{\LARGE Closeness Centrality}} & \makecell{Introduced and formalized\\ temporal centrality definitions\\(including temporal closeness)\\ for time-aggregated \\ settings.} & \makecell{Define temporal shortest\\ paths; compute closeness over\\ time-aggregated windows and\\ time-respecting distances.} & \makecell{Foundational definitions used\\ by many later works;\\ practical for snapshot \\ analyses.} & \makecell{Aggregation/windowing choices\\ can bias results; sensitive\\ to chosen time\\ resolution.} & \makecell{Temporal social and\\ communication networks.} \\
Crescenzi \textit{et al.} \cite{crescenzi2020finding} & &  \makecell{Formalize temporal closeness\\ definitions and propose\\ algorithms for finding\\ top-k nodes for temporal\\ closeness} & \makecell{Harmonic temporal closeness\\ definition;  algorithms that\\ exploit temporal distance\\ properties and pruning\\ to accelerate computation} & \makecell{Handles disconnected/\\partially reachable cases\\ via harmonic definition;\\ practical speedups on\\ real datasets.} & \makecell{Performance depends on\\ temporal sparsity and\\ structure; worst-case\\ still expensive.} & \makecell{Large-scale temporal\\ datasets,  streaming\\ contact networks.} \\
Oettershagen \textit{et al.} \cite{oettershagen2022computing} & & \makecell{Provide efficient exact\\ and heuristic algorithms to\\ compute top-k nodes by\\ temporal closeness} & \makecell{New minimum-duration-path\\ algorithm on temporal graphs plus\\ pruning to compute\\ top-k closeness values\\ efficiently.} & \makecell{Scales much better than\\ all-pairs temporal shortest-path \\ methods for large graphs\\ when only top-k is needed.} & \makecell{Exact computation still\\ costly  for very large,\\ dense temporal graphs.} & \makecell{Large temporal communication\\ networks where identifying\\ few top spreaders is required.} \\ \hline 
Taylor \textit{et al.} \cite{taylor2017eigenvector} &  \multirow{4}{*}{\rotatebox[origin=c]{90}{\LARGE Eigenvector Centrality}} & \makecell{Generalize any eigenvector-based\\ centrality to temporal networks\\ via a joint centrality\\ that captures node x\\ time importance.} & \makecell{Build a supracentrality matrix\\ by coupling centrality matrices\\ of individual time layers;\\ dominant eigenvector\\ gives joint centrality} & \makecell{Principled, linear-algebraic\\ framework; yields joint\\/marginal/conditional\\ centralities and analytic limits\\ for weak/strong coupling} & \makecell{Supracentrality matrix size\\ grows with N*T;\\computational  cost\\ for very large N.} & \makecell{Temporal social, citation,\\ co-authorship, multiplex datasets.}\\
Yin \textit{et al.} \cite{yin2018inter} & & \makecell{Propose an Improved\\ Eigenvector-based Centrality\\ (IECM) that uses inter-layer\\ similarity (coupling strength)\\ to better detect influential\\ nodes} & \makecell{Compute interlayer similarity\\ weights between consecutive\\ layers and couple layer\\ centrality matrices\\ accordingly} & \makecell{Better identification of\\ influential nodes vs. traditional\\ ECM in experiments;\\ accounts for similarity\\ between nearby time layers.} & \makecell{Considers mainly consecutive-\\layer similarity; may miss\\ long-range temporal\\ dependencies.} & \makecell{Dynamic social and\\ communication networks,\\ influence detection.} \\
Taylor \textit{et al.} \cite{taylor2021tunable} & & \makecell{Extend and tune supracentrality\\ via flexible interlayer\\ coupling; unify multiplex\\ and temporal eigenvector\\ centralities.} & \makecell{Introduce general interlayer-adjacency\\ $\tilde{A}$ and coupling strength $\omega$;  use\\ singular-perturbation analysis for\\ weak/strong coupling regimes \\ and compute supracentralities.} & \makecell{Flexible coupling lets one model\\ different temporal coupling patterns;\\ theoretical analysis of coupling\\ limits} & \makecell{More parameters (interlayer topology, $\omega$)\\ add modelling complexity;\\ computation still heavy for large\\ systems.} & \makecell{Multilayer networks across\\ sociology, biology, infrastructure.} \\
Zhao \textit{et al.} \cite{zhao2023general} & & \makecell{Incorporate higher-order \\interactions into supracentrality\\ to capture non-pairwise\\ temporal dependencies.} & \makecell{Build a higher-order\\ supracentrality framework\\ using higher-order adjacency to\\ couple layers and compute \\eigenvector centralities.} & \makecell{Captures richer temporal\\ patterns and higher-order\\ dependencies missed\\ by pairwise coupling.} & \makecell{More complex model and\\ heavier computation; \\ requires higher-order\\ data and careful \\ interpretation.} & \makecell{Temporal systems\\ where higher-order\\ interactions matter} \\ \hline
Rocha \textit{et al.} \cite{rocha2014random} & \makecell{TempoRank} &  \makecell{Ranks influential nodes\\ by preserving the\\ temporal order of\\ interactions.} &\makecell{Computes node importance\\ using a temporal random\\ walk and its stationary\\ distribution.} & \makecell{Captures temporal dynamics\\ and diffusion more effectively\\ than static centrality\\ measures}  & \makecell{Sensitive to temporal\\ resolution and assumes\\ periodic temporal sequences.} & \makecell{Human contact, communication,\\ information diffusion, and\\ epidemic networks.} \\ \hline
Takaguchi \textit{et al.} \cite{takaguchi2016coverage} & \makecell{Temporal Coverage Centrality} &  \makecell{Identifies important temporal\\ vertices based on their\\ participation in the fastest\\ temporal paths.} &\makecell{Measures the fraction of\\ source–destination vertex\\ pairs whose fastest temporal\\ paths pass through a\\ temporal vertex.} & \makecell{Parameter-free, robust to\\ time-scale changes, and\\ efficiently computable using\\ a DAG-based reachability\\ approach.}  & \makecell{Evaluates temporal vertices\\ (node-time pairs) rather than\\ overall node importance,\\ increasing computational\\ complexity for large\\ temporal networks.} & \makecell{Information diffusion,\\influence maximization, and\\ temporal communication \\networks.} \\ \hline
Wang \textit{et al.} \cite{wang2017new} & \makecell{Efficiency centrality} &  \makecell{ To propose a new centrality\\ measure based on the\\ relative change in network\\ efficiency after removing a\\ node} &\makecell{Combines graph efficiency\\ theory with a node-removal\\ strategy, providing a\\ physically interpretable\\ measure of how much each\\ node contributes to overall\\ information transmission efficiency\\ in the network.} & \makecell{Considers both local and\\ global efficiency of\\ information transfer.}  & \makecell{Computationally expensive for\\ large-scale networks\\ due to repeated network\\ efficiency calculations.} & \makecell{Social, transportation,\\ co-authorship, and\\ airline networks} \\ \hline
Oettershagen \textit{et al.} \cite{oettershagen2022temporal} & \makecell{Temporal Walk Centrality} &  \makecell{Identifies influential nodes\\ based on their ability to\\ obtain and disseminate \\information through temporal\\ walks.} &\makecell{Computes node importance by\\ counting weighted temporal\\ walks that respect the\\ chronological order of\\ interactions.} & \makecell{Captures information propagation\\ beyond shortest paths while\\ preserving temporal\\ causality.}  & \makecell{Computationally expensive for\\ large networks, particularly\\ for exact computation using\\ directed line graph\\ expansion.} & \makecell{Information diffusion,\\ communication networks,\\ and fake news monitoring.} \\ 
\bottomrule
\end{tabular}}
\end{table} 
Cruciani et al. \cite{cruciani2024mantra} proposed Mantra, a scalable framework for approximating temporal betweenness centrality using sampling-based techniques. The method efficiently estimates node importance in large temporal networks by sampling time-respecting paths instead of computing all-pairs shortest paths, significantly reducing computational cost. Bub in \cite{buss2020algorithmic} analysed strict, non-strict, foremost, and fastest path variants, highlighting the computational challenges and limitations for exact computation in dynamic networks. Collectively, these works advance the understanding of temporal betweenness, offering both exact and approximate solutions for efficiently identifying influential nodes in evolving networks.\par
Closeness centrality measures the reachability of a node from other nodes in the network. Temporal closeness centrality has been formalised and studied both theoretically and algorithmically to account for time-respecting distances in dynamic networks. Early formalisations and definitions of temporal closeness and time-respecting shortest paths were given by Tang et al. \cite{tang2010analysing}, providing the basis for many later approaches. Crescenzi et al. \cite{crescenzi2020finding} and Oettershagen et al. \cite{oettershagen2022computing} developed efficient algorithms and top-k approaches for computing harmonic temporal closeness on large datasets. Together, these contributions offer foundations, exact constructions, and scalable algorithms for applying temporal closeness in social, communication, transport, and other time-dependent networks. These path-based methods consider temporal ordering, which is important. However, they have two main limitations. First, exact computation is often very expensive for large networks. Second, shortest-path influence does not always match real spreading processes, where diffusion can happen through many parallel paths rather than only shortest ones.\par 
Eigenvector-based methods measure influence using global structure. They have been systematically extended to temporal and multilayer settings through supracentrality and interlayer-coupling frameworks. Taylor et al. \cite{taylor2017eigenvector} introduced a principled supracentrality construction in which centrality matrices from each time layer are coupled into a large supracentrality matrix whose dominant eigenvector yields joint node–time centralities. Yin et al. \cite{yin2018inter} improved this approach by weighting interlayer connections based on similarity. Although this approach considers global structure, it often depends on predefined coupling parameters between layers. The results can change depending on these parameter choices.\par
Beyond eigenvector-based approaches, several temporal centrality measures quantify node importance from different perspectives. TempoRank \cite{rocha2014random} ranks nodes using the stationary distribution of a temporal random walk over ordered snapshots, thereby accounting for the temporal sequence of interactions. Temporal Walk Centrality (TWC) \cite{oettershagen2022temporal} evaluates node importance by counting weighted time-respecting walks that preserve temporal causality, while Temporal Coverage Centrality (TCC) \cite{takaguchi2016coverage} measures the contribution of temporal vertices to the fastest temporal paths between source–destination pairs. Although these methods capture temporal ordering, walk-based connectivity, path coverage, or network robustness, they do not explicitly quantify node importance through changes in the network's global spectral properties, which directly influence information propagation and network connectivity. In addition, most of the studies evaluate these methods structurally, without thorough validation using multiple diffusion models. The reviewed studies in Table \ref{tab:relatedwork}  highlight complementary strengths and weaknesses. 
\subsection{Spectral-Radius based Methods}
Spectral properties of networks have long been associated with diffusion dynamics and epidemic processes. In particular, the largest eigenvalue (spectral radius) of the adjacency matrix plays a fundamental role in determining the epidemic threshold of spreading models. In many diffusion models, the epidemic threshold is inversely proportional to the spectral radius. This means that networks with larger spectral radius are more susceptible to spreading. \par
Spectral-radius-based node importance has been extensively studied in static networks through concepts such as dynamical importance and spectral perturbation analysis. Restrepo \textit{et al.} \cite{restrepo2006characterizing} quantified node importance by measuring the reduction in the largest eigenvalue after node removal, demonstrating its relevance to network dynamics and robustness. Subsequently, Milanese \textit{et al.} \cite{milanese2010approximating} proposed perturbation-based approximations for efficiently estimating spectral changes caused by structural modifications, thereby avoiding repeated eigenvalue computations. Wang et al. \cite{wang2011identifying} quantifies the impact of removing a node on the dominant eigenvalue of the adjacency matrix. Other studies \cite{vukadinovic2013centrality} examined how changes in network structure affect the spectral radius and overall vulnerability.  While these approaches establish the theoretical foundation for spectral-radius-based node importance, they are formulated for static networks or aggregated networks and do not address temporal node ranking or evolving network snapshots. \par
Although the proposed SEC framework is inspired by spectral-radius-based measures developed for static networks, its objective differs fundamentally. Existing methods primarily quantify the structural importance of nodes within a single static graph or estimate spectral changes resulting from structural perturbations. In contrast, SEC extends these concepts to temporal networks by computing spectral node importance independently across successive temporal snapshots and aggregating the resulting importance scores to capture the cumulative influence of nodes over time. More importantly, SEC directly validates node rankings through diffusion simulations under SI, SIS, and IC models. This connects spectral structure with real spreading behavior. The approximate version, ASEC, reduces computational cost while maintaining ranking consistency. Therefore, SEC is not only a spectral centrality measure but a temporally grounded and diffusion-validated centrality framework. Table~\ref{tab:spectral_comparison} summarizes the key methodological differences between representative spectral-radius-based approaches and the proposed framework.
\begin{table}[H]
\renewcommand{\arraystretch}{3} 
\setlength{\tabcolsep}{10pt} 
\caption{Comparison of foundational and representative spectral-radius-based approaches with the proposed SEC framework.}
\label{tab:spectral_comparison}
\resizebox{\textwidth}{!}{  
\begin{tabular}{cccccc}

\toprule
\textit{Criterion} &

\textbf{Restrepo et al. \cite{restrepo2006characterizing}} &
\textbf{Milanese et al. \cite{milanese2010approximating}} &
\textbf{Wang et al. \cite{wang2011identifying}} &
\textbf{Vukadinovic et al. \cite{vukadinovic2013centrality}} &
\textbf{Proposed SEC / ASEC} \\
\midrule

\textbf{Objective}
&
\makecell{Node importance}
&
\makecell{Approximate spectral impact\\ of structural perturbations}
&
\makecell{Identify structurally important\\ nodes using spectral\\ node removal}
&
\makecell{Analyse spectral radius\\ in dynamic communication\\ networks}
&
\makecell{Rank influential nodes in\\ temporal networks using \\ spectral efficiency}
\\

\textbf{Network Type}
&
\makecell{Static}
&
\makecell{Static}
&
\makecell{Static}
&
\makecell{Dynamic communication networks}
&
\makecell{Temporal snapshot-based\\ networks}
\\

\textbf{Primary Focus}
&
\makecell{Node importance}
&
\makecell{Efficient spectral\\ approximation}
&
\makecell{Community-aware node\\ ranking}
&
\makecell{Network-level spectral\\ analysis}
&
\makecell{Temporal node ranking}
\\

\textbf{Node Removal}
&
\makecell{Yes}
&
\makecell{Yes}
&
\makecell{Yes}
&
\makecell{No}
&
\makecell{Yes}
\\

\textbf{Spectral Radius Utilization}
&
\makecell{Exact eigenvalue \\reduction}
&
\makecell{Perturbation-based \\estimation}
&
\makecell{Exact eigenvalue \\recomputation}
&
\makecell{Spectral radius\\ monitoring}
&
\makecell{Exact SEC and approximate\\ ASEC}
\\

\textbf{Approximation Mechanism}
&
\makecell{First-order perturbation}
&
\makecell{First- and higher-order
\\perturbation}
&
\makecell{None}
&
\makecell{None}
&
\makecell{First-order perturbation (ASEC)}
\\
\textbf{Temporal Snapshot Analysis}
&
\makecell{No}
&
No
&
No
&
No
&
Yes
\\

\textbf{Temporal Aggregation}
&
No
&
No
&
No
&
No
&
Yes
\\
\textbf{Diffusion-based Validation}
&
No
&
No
&
No
&
\makecell{Epidemic threshold\\ analysis}
&
\makecell{SI, SIS and IC diffusion\\ models}
\\

\textbf{Theoretical Analysis}
&
\makecell{Perturbation analysis}
&
\makecell{Perturbation analysis}
&
\makecell{Spectral node-removal\\ analysis}
&
\makecell{Spectral radius\\ analysis}
&
\makecell{Approximation error \\analysis}
\\

\bottomrule
\end{tabular}}
\end{table}


\section{Preliminaries}
\label{sec:Preliminaries}
In this section, we will discuss the preliminaries related to temporal networks and various diffusion models used to evaluate the effectiveness of our proposed methods.
\subsection{Definitions}
\begin{Def}(Temporal Network)
A network where interactions among nodes evolve over time can be represented using time-varying graphs. A contact between two nodes $x, y \in N$ is described by a quadruplet $c = (x, y, t, \delta t)$, where $0 \leq t \leq T$ denotes the start time of the interaction and $\delta t$ represents its duration, measured in appropriate temporal units \cite{nicosia2013graph}. 
\end{Def}
\begin{Def}[Spectral Radius]
Let $A_t$ denote the adjacency matrix of the temporal snapshot $G_t$. The \emph{spectral radius} of $A_t$, denoted by $\lambda(A_t)$, is defined as:
\begin{center}
$\lambda(A_t)=\max_i |\lambda_i(A_t)|$,
\end{center}
where $\lambda_i(A_t)$ are the eigenvalues of $A_t$. Since each temporal snapshot is modeled as an undirected graph, $A_t$ is symmetric and all its eigenvalues are real. Hence, the spectral radius is simply the largest eigenvalue of $A_t$. It serves as a global structural descriptor and is widely used to characterize network connectivity, diffusion dynamics, and robustness.
\end{Def}
\begin{Def}(Largest Connected Component (LCC):)
For a graph $G=(V,E)$, the Largest Connected Component (LCC) is the connected component containing the maximum number of vertices. Its size is given by
\[
LCC(G)=\max_{C_i\in\mathcal{C}(G)} |C_i|,
\]
where $\mathcal{C}(G)$ is the set of connected components of $G$. 
\end{Def}
\begin{Def}(Perron-Frobenius Theorem)
It states that a  non-negative irreducible matrix has a unique, positive, simple eigenvalue (the Perron-Frobenius eigenvalue) equal to its spectral radius, and the associated eigenvector is strictly positive \cite{perron2007perron}.
    
\end{Def}
\begin{Def}(Network Efficiency)
The network efficiency $E(G)$ is the average of the inverse of the shortest path length $d_{ij}$ \cite{latora2001efficient} and is given as: 
\begin{center}
    $E(G)=\frac{1}{N*(N-1)}\sum_{i \neq j \in G} \frac{1}{d_{ij}}$
\end{center}
\end{Def}
\begin{Def}(Efficiency Centrality (EffC))
It is a method used to assess the importance of an individual node across the entire network. It finds important nodes by removing them one at a time and seeing how much the overall efficiency of the network drops. The EffC of node $v$ is given as \cite{wang2017new}:
\[
    C_{Effc}(v)=\frac{E(G)-E(G')}{E(G)}, v \in V
    \]

Here, $E(G')$ gives the efficiency centrality of the network $G'$ when node $v$ is removed.
\end{Def}
\subsection{Diffusion Models}
To evaluate the spreading capability of nodes selected by different centrality measures, we employ three widely used diffusion models on temporal networks: Susceptible–Infected (SI), Susceptible–Infected–Susceptible (SIS), and Independent Cascade (IC). These models capture different epidemic and information propagation dynamics and enable a comprehensive assessment of influence in evolving networks.\par
Let $G_t = (V_t, E_t)$ denote the temporal network snapshot at time $t$, where $V_t$ is the set of nodes present in snapshot $t$ and $E_t$ represents the edges active during that time window.
\begin{enumerate}
    \item {Susceptible–Infected (SI) Model:}
In the SI diffusion model, nodes exist in one of two possible states: susceptible $(S)$ or infected $(I)$. During each time step, an infected node spreads the infection to its susceptible neighbors with transmission probability $\beta$. After infection, a node permanently remains in the infected state throughout the diffusion process.
\item {Susceptible–Infected–Susceptible (SIS) Model:}
In this model, nodes alternate between susceptible and infected states without permanent immunity. Infection spreads with probability $\beta$, and infected nodes recover with probability $\gamma$, returning to the susceptible state. 
\item {Independent Cascade (IC) Model:}
This model is a discrete-time stochastic process. At time $t$, newly infected nodes attempt to activate each currently inactive neighbor with probability $p$. Each activation attempt is performed only once, and all attempts are independent.
The process continues iteratively until no further activations occur.
\end{enumerate}
\section{Proposed Method}
The structural integrity of a network is strongly influenced by nodes whose removal significantly alters its global connectivity. In particular, variations in the spectral radius of the adjacency matrix provide a principled way to assess changes in diffusion capacity and overall robustness. Despite this, many existing temporal centrality measures emphasize path-based or locally aggregated properties, without directly quantifying the global spectral impact of node failures. \par
To address this limitation, we propose Spectral Efficiency Centrality (SEC), a temporal centrality framework that evaluates node importance through spectral vitality. Specifically, SEC measures the relative change in the network’s spectral radius induced by node removal at each temporal snapshot. By aggregating these effects across time, the proposed metric identifies nodes that consistently contribute to maintaining global structural cohesion and diffusion potential in evolving networks.

\label{sec:proposedmethod}
\subsection{Spectral Efficiency Centrality (SEC)}
SEC is developed by extending the concept of Efficiency Centrality (EffC) \cite{wang2017new} and spectral vitality\cite{restrepo2006characterizing}\cite{milanese2010approximating}. In traditional Efficiency Centrality, node importance is quantified by the reduction in the overall network efficiency after removing a node, thereby emphasizing pairwise shortest-path connectivity. In contrast, SEC measures node importance through the normalized reduction in the spectral radius of the network after node removal, allowing the influence of a node to be assessed from the perspective of global network connectivity and diffusion capability. While spectral-radius-based node importance has been previously investigated for static networks through concepts such as dynamical importance and spectral perturbation analysis, the proposed SEC extends these principles to temporal networks by evaluating spectral importance across successive temporal snapshots and aggregating the resulting scores to characterize the cumulative structural influence of nodes over time. Since the spectral radius is closely related to epidemic thresholds and spreading dynamics, this formulation provides a global perspective on node influence while remaining computationally tractable through the proposed ASEC approximation.

\par
The proposed SEC framework adopts a snapshot-based representation of temporal networks. The temporal edge stream is first chronologically ordered according to timestamps and partitioned into successive time windows. Rather than relying solely on fixed temporal segmentation, structurally meaningful snapshots are identified through a Jaccard similarity-based change detection procedure applied to consecutive windows. Consequently, the temporal ordering of network evolution is preserved while each snapshot represents a distinct structural state of the evolving network. To understand this, let a temporal network be represented as a sequence of undirected graph snapshots $G_1, G_2,\dots, G_n$, where each $G_t=(V_t, E_t)$ denotes the network at time $t$. Let $A_t$ denote the adjacency matrix of the snapshot graph $G_t$.  We define the spectral efficiency of a node $v$  at snapshot $t$ as:
\begin{equation}
    C^{(t)}_{SEC}(v)=\frac{\lambda(A_t)-\lambda(A_t- v)}{\lambda(A_t)},\label{eq:sec}
\end{equation}
Here,  $\lambda(A - v)$ is the spectral radius of the snapshot graph after node $v$ is removed.
The overall Spectral Efficiency Centrality of node $v$  across the temporal network is:
\begin{equation}
    C_{SEC}(v)=\frac{\sum_{t \in {T_v}}C^{(t)}_{SEC}(v) }{T},
\end{equation}
Here, $T$ is the total number of snapshots.\par
Unlike existing centrality metrics, SEC captures the global importance of a node by evaluating how its removal affects the network’s spectral properties, effectively identifying critical spreaders whose presence is vital for information or epidemic propagation. Its temporal formulation considers the persistence and timing of interactions across snapshots, making it particularly suitable for evolving networks such as email communications, social media, or transportation systems. \par
Beyond its intuitive interpretation, a centrality measure must satisfy basic theoretical requirements to be meaningful and comparable across networks. We therefore analyze several fundamental properties of SEC, including boundedness and its response to structural perturbations. These properties ensure that SEC remains numerically stable, interpretable, and consistent across different network topologies and sizes.
\subsubsection{Theoretical Properties:}
\label{subsec:theoriticalprop}
We discuss several fundamental properties of the proposed Spectral Efficiency Centrality (SEC). These properties justify its mathematical soundness and its interpretation as a spectral perturbation-based measure. \par
Let $A$ be the adjacency matrix of a connected undirected graph $G$ with spectral radius $\lambda(A)$ and corresponding Perron eigenvector $x$ satisfying $Ax = \lambda(A)x$, where $x_i > 0$ for all $i$.\par
Property 1 - Boundedness of SEC:
For any temporal snapshot $G_t$ and any node $v \in V(G_t)$, the snapshot-level spectral efficiency satisfies $0 \leq C^{(t)}_{SEC}(v) \leq 1$ and consequently the aggregated $C_{SEC}(v) \in  [0,1]$. Since node removal from an undirected graph cannot increase the spectral radius of an undirected graph, the numerator is non-negative, and normalisation by $\lambda(A)$ ensures the upper bound. We have $0 \le \lambda(A) - \lambda(A - v) \le \lambda(A)$, which implies $0 \le C_{SEC}(v) \le 1$. \par
Property 2 - Global Sensitivity:
SEC assigns significantly higher scores to structurally central nodes whose removal induces large reductions in the spectral radius. In particular, nodes with high connectivity or strong participation in dense substructures (e.g., hubs or core nodes) attain higher $C_{SEC}$ values.\par
Property 3 - Robustness to peripheral nodes: 
SEC is robust to the removal of peripheral or weakly connected nodes. Nodes that contribute little to the global structure of the network induce only negligible changes in the spectral radius when removed, resulting in $C_{SEC}$ values close to zero. We empirically validate these properties using controlled experiments on synthetic networks in section \ref{sec:empericalvalidation}.
\subsubsection{Computational Bottleneck:}
While SEC effectively captures the global structural importance of nodes, its reliance on repeated eigenvalue computations makes it computationally 
expensive for large or frequently evolving networks. For a temporal snapshot with $n$ nodes, computing $C_{SEC}$ requires evaluating the spectral radius of the network after the removal of each node. This results in $O(n)$ spectral radius computations per snapshot.\par
In the worst case of dense graphs, each spectral radius computation requires a full eigenvalue decomposition of an $n \times n$ adjacency matrix, which has time complexity $O(n^3)$. Consequently, the overall worst-case computational complexity of SEC is $O(n^4)$ per snapshot. When extended to temporal networks with $T$ snapshots, the total complexity scales as $O(T n^4)$, rendering exact computation infeasible for large networks.\par
For sparse graphs, iterative methods such as the power method \cite{golub2013matrix} or Lanczos algorithm \cite{lanczos1950iteration} can be employed to compute the leading eigenvalue more efficiently. In this case, each spectral radius computation typically costs $O(m)$ per iteration, where $m$ denotes the number of edges, and the total cost per computation is $O(km)$ for $k$ iterations. However, since this procedure must still be repeated for each node removal, the total cost per snapshot becomes $O(nm)$, which reduces to $O(n^2)$ for sparse networks where $m = O(n)$. Despite this improvement, the repeated spectral recomputation across nodes and snapshots remains computationally expensive in large-scale temporal settings.\par
To address this issue, we develop an efficient approximation of SEC based on first-order eigenvalue perturbation theory. The proposed approximation avoids repeated spectral recomputation and requires computing the leading eigenpair only once per snapshot. It significantly reduces the computational overhead while preserving the relative influence patterns captured by SEC, enabling scalable analysis of real-world temporal networks.
\subsection{Approximate Spectral Efficiency Centrality (ASEC)}
\label{subsec:ASEC}
The proposed approximation (ASEC) is based on classical first-order eigenvalue perturbation theory for symmetric matrices \cite{stewart1990matrix} \cite{kato1966perturbation}, which states that small, localised changes to a matrix induce changes in its dominant eigenvalue that are proportional to a quadratic form involving the leading eigenvector. Throughout the analysis, each temporal snapshot is represented by an undirected and unweighted graph, resulting in a symmetric adjacency matrix. Furthermore, the dominant eigenvalue is assumed to be simple, ensuring the applicability of first-order eigenvalue perturbation theory. Since the removal of a single node typically introduces a localized structural modification, the resulting perturbation is assumed to be sufficiently small such that higher-order perturbation terms can be neglected.\par
Let $A_t$ denote the adjacency matrix of a snapshot graph with spectral radius $\lambda(A_t)$ and associated Perron-Frobenius eigenvector $x$, normalized such that $\sum_{i} x_i^2 = 1$. For a localized perturbation $\Delta A_t$, the resulting change in the spectral radius satisfies
\[
\Delta \lambda \approx \frac{x^\top (\Delta A_t)x}{x^\top x}.
\] 
The above expression represents the first-order perturbation of the dominant eigenvalue caused by a localized modification of the adjacency matrix. Since the dominant eigenvector captures the principal structural characteristics of the network, the first-order approximation provides an efficient estimate of the change in spectral radius without requiring complete eigendecomposition after each node removal.\par
In the case of node removal, the perturbation is concentrated around the incident edges of the removed node, implying that the corresponding reduction in the spectral radius is approximately proportional to the squared entry of the leading eigenvector associated with that node ($x_v$). Although node removal is not strictly an infinitesimal perturbation, first-order eigenvalue perturbation provides an efficient approximation for estimating the corresponding spectral change and has been widely adopted in spectral node importance and network robustness studies \cite{restrepo2006characterizing}\cite{milanese2010approximating}.\par
\[
\Delta \lambda \approx  x_v^2.
\] 
For adjacency matrices of connected graphs, Perron–Frobenius theory guarantees that the spectral radius is simple and associated with a strictly positive eigenvector \cite{meyer2023matrix}, ensuring the validity of the perturbation expansion.

\[
\lambda(A_t - v) \ge \lambda(A_t) - \lambda(A_t)\, x_v^2
\]
Substituting the approx value of $\lambda(A_t - v)$ in Eq. \eqref{eq:sec}, we get the Approximate Spectral Efficiency Centrality (ASEC) of a node $v$ at snapshot $t$ as:
\[
\mathrm{C^{(t)}_{ASEC}}(v) \approx \frac{\lambda(A_t) - ( \lambda(A_t) - \lambda(A_t)\, x_v^2)}{\lambda(A_t)} \approx x_v^2 \]
ASEC is intended as an efficient approximation to SEC rather than an exact replacement. The validity of this approximation is further assessed through theoretical and empirical analyses presented in Section \ref{sec:expeval}.
Importantly, the computation of $x$ and $\lambda(A_t)$ is already required for the $C_{SEC}$ formulation. Therefore, once the leading eigenpair $(\lambda(A_t), x)$ is computed for a snapshot, the $C_{ASEC}$ value of all nodes can be obtained directly from the squared components $x_v^2$ without performing additional eigenvalue computations. 
The overall Approximate Spectral Efficiency Centrality (ASEC) of node $v$  across the temporal network is:
\begin{equation}
    C_{ASEC}(v) \approx \frac{\sum_{t \in {T_v}} C^{(t)}_{ASEC}(v)}{T},
\end{equation}
Here, $T$ is the total number of snapshots.
The proposed approximation ASEC requires only a single eigenvalue–eigenvector computation, reducing the computational cost by a factor of $n$ while preserving the relative importance of nodes. The first-order approximation provides an efficient estimate of SEC by retaining only the dominant perturbation term associated with node removal. Consequently, the approximation error depends exclusively on the neglected higher-order perturbation terms. Proposition 1 formalizes this relationship by expressing the approximation error as a normalized higher-order remainder term.\par
\begin{proposition}[Approximation Error]
Let $G=(V,E)$ be a connected graph with adjacency matrix $A$.
Let $\lambda(A)$ be the spectral radius of $A$, and let
$\mathbf{x}$ be the associated Perron--Frobenius eigenvector
normalized such that $\|\mathbf{x}\|_2^2 = 1$.

For a node $v \in V$, we define the Spectral Efficiency Centrality as
\[
\mathrm{C_{SEC}}(v) = \frac{\lambda(A) - \lambda(A - v)}{\lambda(A)},
\]
and its Perron--Frobenius approximation as
\[
\mathrm{C_{ASEC}}(v) \approx x_v^2.
\]

Then the absolute approximation error
\[
\varepsilon(v)=
|\mathrm{C_{SEC}}(v) - \mathrm{C_{ASEC}}(v)|
\]
satisfies
\[
\varepsilon(v)
=
\frac{R_v}{\lambda(A)}
\]
where $R_v \ge 0$ is the higher-order remainder term arising from
eigenvalue perturbation.
\end{proposition}
\begin{proof}
Let $\Delta A$ denote the perturbation matrix corresponding to the deletion of node $v$ and all its incident edges, such that
\[
\Delta \lambda \approx \frac{x^\top (\Delta A)x}{x^\top x}.
\] 
From first-order eigenvalue perturbation theory for symmetric matrices,
the dominant eigenvalue satisfies:
\begin{equation}
\lambda(A - v)
=
\lambda(A)
-
x^\top \Delta A x
-
R_v, \label{eq:error}
\end{equation}
where $R_v \ge 0$ represents higher-order terms neglected by the
first-order approximation.

For node removal, the quadratic form reduces to
\[x^\top \Delta A x
=
\lambda(A)\,x_v^2.
\]

Substituting this in Eq. \eqref{eq:error}, we get
\[
\lambda(A - v)
=
\lambda(A)
-
\lambda(A)x_v^2
-
R_v.
\]
Dividing both sides by $\lambda(A)$ gives the exact SEC:
\[
\mathrm{C_{SEC}}(v) = x_v^2 + \frac{R_v}{\lambda(A)}.
\]

Since $\mathrm{C_{ASEC}}(v) \approx x_v^2$, the absolute
approximation error is
\[
|\mathrm{C_{SEC}}(v) - \mathrm{C_{ASEC}}(v)|=
\frac{R_v}{\lambda(A)}.
\]
\end{proof}
This shows that the approximation error depends solely on the normalized higher-order remainder term $R_v$. Therefore, the accuracy of ASEC is expected to improve whenever the contribution of higher-order perturbation terms is negligible. The higher-order remainder term $R_v$ is not directly computable, as it represents the neglected higher-order terms of the eigenvalue perturbation expansion. Therefore, its practical effect and theoretical expectation are examined empirically in Section \ref{subsec:Emeracc} through approximation error analysis and ranking consistency evaluations.
\section{Experimental Evaluation} 
\label{sec:expeval}
To comprehensively evaluate the proposed SEC and ASEC framework, the experimental protocol examines four complementary aspects: (i) diffusion capability using SI, SIS, and IC models; (ii) ranking fidelity of the proposed approximation (ASEC) through approximation error, Top-k overlap, and Kendall's $\tau$; (iii) computational efficiency through runtime comparisons; and (iv) robustness analysis using spectral radius and largest connected component degradation under targeted node removal. Together, these experiments evaluate the effectiveness of SEC from temporal, diffusion, computational, and structural perspectives.
\subsection{Setup} The algorithms were implemented in Python 3.11 using the NumPy, SciPy, NetworkX, Pandas, and Matplotlib libraries. The experiments were conducted on Google Colab using CPU-based computation with approximately 16 GB RAM. For all diffusion experiments, 10000 Monte Carlo simulations were performed using a fixed random seed (42) to ensure reproducibility. The implementation parameters for each dataset are summarised in Table \ref{tab:topologySD} and \ref{tab:topologyRW}.
\subsubsection{Datasets}
\label{sec:datasets}
\begin{enumerate}
\item {Synthetic Datasets:} To empirically validate the theoretical properties of SEC discussed in \ref{subsec:theoriticalprop}, we conducted controlled experiments on synthetic temporal networks generated using Erdős–Rényi, Barabási–Albert, and Watts–Strogatz models.  Each configuration is averaged over 20 independent runs.  Table \ref{tab:topologySD} shows the basic topological and key properties of the three models.
\begin{enumerate}
    \item {Erdős–Rényi(ER) model:} It generates a random graph $G(n,p)$, with $n$ number of nodes, where $p$ is the probability that any pair of distinct nodes is connected by an edge.
    \item {Barabási–Albert(BA) model:} It creates networks that grow by preferential attachment. Nodes are added one by one, and new nodes prefer to connect to $m$ already well-connected nodes. 
    \item{Watts–Strogatz(WS) model:} It produces a small-world network $G(n,k,\beta)$, where $n$ is the number of nodes, $k$ is the number of nearest neighbors each node is initially connected to in a regular ring lattice, and $\beta$ is the probability of rewiring each edge.
\end{enumerate}
\begin{table}[H]
\renewcommand{\arraystretch}{0.95} 
\setlength{\tabcolsep}{10pt} 
\centering
\caption{Basic topological characteristics of synthetic network models}
\label{tab:topologySD}
\resizebox{\textwidth}{!}{
\begin{tabular}{lcccccc}
\toprule
\textbf{Model} & \textbf{Type} & \textbf{Key Feature} & \textbf{Clustering} & \textbf{Degree Distribution} & \textbf{$n$} & \textbf{Parameter} \\
\midrule
\multirow{3}{*}{ER} & \multirow{3}{*}{Random} &  \multirow{3}{*}{Purely random edges} &  \multirow{3}{*}{Low} &  \multirow{3}{*}{Poisson (No hubs)} & 50  & \multirow{3}{*}{$p=0.1$} \\
& & & & & 100 &  \\
& & & & & 200 & \\
\hline

\multirow{3}{*}{BA} & \multirow{3}{*}{Scale-Free}& \multirow{3}{*}{Growth and Preferential Attachment} & \multirow{3}{*}{Low/Moderate} &  \multirow{3}{*}{Power -law (Hubs)} & 50  & \multirow{3}{*}{$m=3$}   \\
& & & & & 100 & \\
& & & & & 200 &  \\

\hline
\multirow{3}{*}{WS}& \multirow{3}{*}{Small-World} & \multirow{3}{*}{Shortest path length} & \multirow{3}{*}{High} &\multirow{3}{*}{Regular to Random} & 50  & \multirow{3}{*}{$k=4,\ p=0.1$}   \\
& & & & &  100 &   \\
& & & & &  200 &   \\
\bottomrule
\end{tabular}}
\end{table}
\item {Real-world Temporal Datasets: }To systematically assess the performance of the proposed SEC and ASEC measure, we conducted experiments on three real-world temporal network datasets. A comprehensive description of the three selected networks is presented below. Table \ref{tab:topologyRW} gives the statistical overview of the three networks.
\begin{enumerate}
\item Enron email communication network \footnote{https://snap.stanford.edu/data/email-Enron.html}: It covers all the email communication within a dataset of around half million emails.  Nodes of the network are email addresses and if an address $x$ sent at least one email to address $y$, the graph contains an undirected edge from $x$ to $y$. 
\item CollegeMsg temporal network \footnote{https://snap.stanford.edu/data/CollegeMsg.html}: This dataset consists of private messages exchanged on an online social network at the University of California, Irvine. Each temporal edge (u,v,t) represents a private message sent from user u to user v at time t.
\item Math Overflow temporal network \footnote{https://snap.stanford.edu/data/sx-mathoverflow.html} : This dataset represents a temporal interaction network from the Stack Exchange website Math Overflow. A directed edge $(x,y,t)$ denotes an interaction occurring at time $t$, where user $x$ interacts with user $y$. Since SEC is formulated for undirected temporal networks, the directed interactions were converted into undirected edges during snapshot construction by ignoring edge directions.

\end{enumerate} 
\begin{table}[H]
\renewcommand{\arraystretch}{0.95} 
\setlength{\tabcolsep}{10pt} 
\centering
\caption{Statistical Overview of the three real-world networks.}
\label{tab:topologyRW}
\resizebox{\textwidth}{!}{
\begin{tabular}{lcccccc}
\toprule
\textbf{Network} & \textbf{Type} & \textbf{$n$} & \textbf{$m$} &\textbf{Snapshots} & \textbf{Time-window(in days)} & \textbf{Seed Nodes}\\
\midrule
Enron Email & Moderate-sized & 986  & 332334 & 12 & 7 & 5  \\ 
CollegeMsg & Moderate-sized & 1899  & 59835 & 19 & 10 & 10  \\ 
Math Overflow & Large and sparse & 24818 & 506550 & 27 & 90 & 15 \\
\bottomrule
\end{tabular}}
\end{table}
\textit{Time-window:} The time-window sizes were selected according to the scale, interaction frequency, and temporal characteristics of each dataset to achieve an appropriate balance between temporal resolution and structural stability. For the Enron Email dataset, a 7-day window was adopted to capture weekly communication patterns while avoiding excessively sparse daily snapshots. The CollegeMsg dataset exhibits relatively frequent online interactions over a longer observation period; therefore, a 10-day window was chosen to preserve the temporal evolution of user communications while maintaining sufficient interaction density within each snapshot. For the large-scale MathOverflow network, a coarser 90-day window was employed to ensure adequate edge density in each temporal layer and to prevent excessive fragmentation resulting from sparse interactions. These dataset-specific choices preserve meaningful temporal dynamics while maintaining structurally informative snapshots suitable for temporal centrality analysis.\par
\textit{Seed-size:} \textit{Seed-size:} For each centrality measure, the top-$k$ ranked nodes were selected as the initial seed set for the diffusion experiments. The seed size was chosen as approximately $0.5\%$ of the total number of nodes in each network to provide a fair and consistent evaluation across datasets of different sizes. Accordingly, 5 and 10 seed nodes were selected for the Enron Email and CollegeMsg datasets, respectively. For the large MathOverflow network, a smaller seed set of 15 nodes was deliberately adopted so that the diffusion spreads properly without quickly saturating the
network. This helps ensure a fair comparison between methods.

\end{enumerate}
\subsubsection{Baseline Centrality Measures} We computed SEC and ASEC for every temporal snapshot, and compared it against several well-known temporal centrality measures given below.
\begin{enumerate}
\item {Temporal Walk Centrality(TWC):} It evaluates node importance by modeling information flow through time-respecting random walks in temporal networks. Unlike shortest-path-based measures, it captures dynamic information propagation over feasible temporal paths, enabling effective identification of influential nodes \cite{oettershagen2022temporal}.
\item{Efficiency Centrality (EffC):} It is introduced to rank spreaders across the entire network, identifying influential nodes by evaluating the change in overall network efficiency when each node is removed \cite{wang2017new}.
\item {Supracentrality:} It is a systematic extension of centrality measures that applies to any eigenvector-based metric. Specifically, a temporal network with $N$ nodes is modelled as a sequence of $T$ layers, each representing the network within a given time window. The centrality matrices of these layers are then combined into a single supracentrality matrix of dimension $NT \times NT$, whose dominant eigenvector yields the centrality score of each node $i$ at each time 
$t$ \cite{taylor2017eigenvector}. 
\item {Temporal Coverage Centrality(TCC):} It measures the importance of temporal vertices based on their participation in the fastest time-respecting paths within a temporal network. It provides a parameter-free and efficient way to identify critical temporal vertices that facilitate information flow \cite{takaguchi2016coverage}.
\item {TempoRank:} It is a temporal network centrality measure that extends random-walk-based ranking to dynamic networks by explicitly incorporating the temporal ordering of interactions. It computes node importance from the stationary distribution of random walkers over time, capturing both structural connectivity and temporal dynamics \cite{rocha2014random}.
\item {Temporal Coreness:} It is an algorithm for maintaining the core number of each node in a temporal graph. The key insight is that when an edge is added or removed, only a subset of nodes is affected and therefore requires core number updates \cite{li2013efficient}.
\item {CENDY:} It is an efficient method for dynamically updating Closeness centrality and avErage path leNgth in Dynamic Networks undergoing edge insertions or deletions \cite{yen2013efficient}.
\item {QUBE:} It is a technique that narrows the search space by identifying a set of candidate vertices whose betweenness centrality values are affected after edge addition or removal \cite{lee2012qube}.
\end{enumerate}
We have discussed these measures in detail in Section \ref{sec:relatedW}.
\subsubsection{Evaluation Metrics} To comprehensively evaluate the effectiveness of the proposed centrality measures, multiple complementary evaluation criteria are employed to assess their diffusion capability, approximation quality, structural importance, and computational efficiency.
\begin{enumerate}
\item{Area Under the Infection Curve (AUC):} For each diffusion model, we compute the AUC of the infected proportion over time. Let $L(t)$ denote the fraction of infected nodes at time $t$. The AUC is defined as: $AUC = \sum_{t=1}^{T} L(t)$. Higher AUC values indicate faster diffusion, larger outbreak size, and stronger global influence of selected seed nodes. Unlike visual curve comparison, we report AUC values in tabular form for all datasets and diffusion models to enable precise quantitative comparison.
\item{Kendall Rank Correlation ($\tau$):} To assess ranking consistency between SEC and ASEC, we compute Kendall’s $\tau$ correlation coefficient:
$\tau = \frac{C - D}{\binom{n}{2}}$, 
where $C$ = concordant pairs and $D$ = discordant pairs. High $\tau$ confirms that ASEC preserves the ranking quality of SEC.
\item{Robustness Analysis:} Since SEC quantifies the contribution of temporal nodes to the global spectral structure of the network, its ability to identify structurally critical nodes is further evaluated through robustness analysis. Nodes are removed sequentially according to the rankings produced by each centrality measure, and the resulting degradation of the network is assessed using two complementary metrics:
\begin{itemize}
    \item {Normalized spectral radius:} It measures the reduction in global spectral connectivity; and
    \item {Normalized Largest Connected Component (LCC):} It quantifies the fragmentation of the network after targeted node removal.
\end{itemize}
Centrality measures that produce a faster decline in these quantities are considered more effective at identifying structurally important nodes.
\item{Computational Efficiency:} Since SEC is spectral-based,  computational efficiency is an important performance criterion. The average execution time of SEC, ASEC, and the competing centrality measures is recorded over all temporal snapshots to evaluate their scalability for large temporal networks.

\end{enumerate}
\subsection{Empirical Validation of Theoretical Properties of SEC}
\label{sec:empericalvalidation}
To empirically validate the theoretical properties of SEC, synthetic networks discussed in Section \ref{sec:datasets} of sizes $N \in \{50, 100, 200\}$ were generated, and 20 independent realizations were considered per configuration, resulting in a total of $180$ network snapshots. For each snapshot, the spectral radius was computed before and after the removal of every node, and the corresponding $C_{SEC}$ values were recorded. To ensure numerical stability, all spectral computations were performed on the largest connected component using symmetric eigensolvers. 
\begin{table}[H]
\renewcommand{\arraystretch}{0.95} 
\setlength{\tabcolsep}{10pt} 
\centering
\caption{Minimum and maximum $C_{SEC}$ values observed across synthetic network models.}
\label{tab:sec_boundedness}
\begin{tabular}{lccc}
\hline
\textbf{Model} & \textbf{Nodes} & \textbf{Min $C_{SEC}$} & \textbf{Max $C_{SEC}$} \\
\hline
BA & 50  & 0.000475 & 0.285646 \\
BA & 100 & 0.000126 & 0.198202 \\
BA & 200 & 0.000018 & 0.199395 \\
ER & 50  & 0.000057 & 0.094842 \\
ER & 100 & 0.000041 & 0.043521 \\
ER & 200 & 0.000605 & 0.021012 \\
WS & 50  & 0.000235 & 0.064189 \\
WS & 100 & 0.000040 & 0.038360 \\
WS & 200 & 0.000018 & 0.029440 \\
\hline
\end{tabular}
\end{table}
Table \ref{tab:sec_boundedness} reports the minimum and maximum $C_{SEC}$ values observed across all experimental runs. Across more than 18,000 individual node removals, all observed $C_{SEC}$ values satisfied $0 \leq C_{SEC}(v) \leq 1$, with no boundedness violations detected. Barabási–Albert(BA) models have relatively higher Max $C_{SEC}$ values, which reflect the stronger influence of hub nodes in these models, i.e., removing a hub causes a large spectral drop. This indicates that SEC is largely insensitive to peripheral or weakly connected nodes, while selectively emphasizing structurally significant vertices that meaningfully affect the global spectral properties of the network. These results empirically confirm the boundedness and robustness of SEC and establish its numerical stability across diverse network topologies.\par
To directly assess the global sensitivity of SEC, we first compare $C_{SEC}$ magnitudes across different network models. Figure \ref{fig:distribution} illustrates the distribution of $C_{SEC}$ values on a logarithmic scale. The differences in $C_{SEC}$ distributions across network models arise from how these networks are built. In BA networks, a few hub nodes accumulate many connections and play a dominant role in maintaining the overall connectivity of the network. Removing such hubs leads to a huge reduction in the spectral radius, which explains the presence of large $C_{SEC}$ values and the long tail observed in the BA model. In ER and WS models, removing any single node has a relatively small effect on the global structure, leading to $C_{SEC}$ values that are concentrated near zero. These observations  empirically confirm that SEC selectively emphasizes structurally critical nodes rather than uniformly distributing importance across the network.
\begin{figure}[H]
    \centering
     \caption{Log-density distributions of SEC across ER, BA, and WS networks.}
    \includegraphics[width=0.48\linewidth]{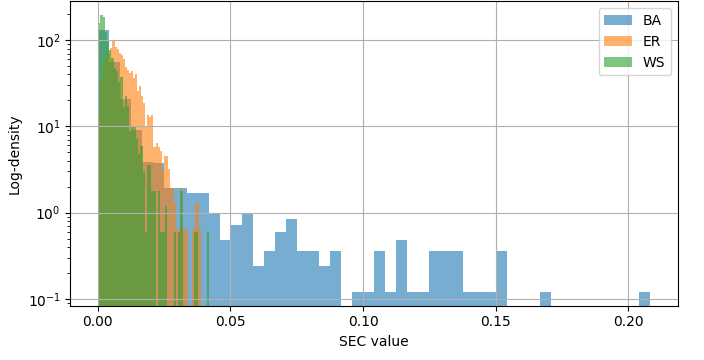}
   
    \label{fig:distribution}
\end{figure}


\subsection{{Runtime Comparison of ASEC with Baseline Centrality Measures}}
We compare the log-scale runtime comparison of SEC, its approximation ASEC, and representative baseline temporal centrality measures (TWC, TCC, tempoRank, CENDY, QUBE, temporal coreness, supracentrality and EffC) across network snapshots of real-world datasets discussed in Section \ref{sec:datasets}. In figure \ref{fig:runtime_plot_EE}, \ref{fig:runtime_plot_CM}, and \ref{fig:runtime_plot_MO}, each data point corresponds to a single snapshot, where the x-axis represents the snapshot index and the y-axis shows the execution time on a logarithmic scale. The logarithmic scale reveals pronounced differences in computational cost across methods. The runtime of SEC grows rapidly with increasing network size, reflecting the repeated eigenvalue computations required for each node removal. In contrast, ASEC remains in the millisecond range across all snapshots, making it highly scalable for real-time and large-scale temporal networks.\par
Efficiency Centrality (EffC), while effective on smaller datasets, exhibited severe computational scalability limitations on the MathOverflow dataset. The execution time increased dramatically with network size and temporal complexity, and the method did not terminate within practical runtime constraints. Therefore, EffC was excluded from the runtime comparison in Figure \ref{fig:runtime_plot_MO}. This observation highlights the limited scalability of EffC for large-scale temporal networks. In contrast, SEC and ASEC maintained stable runtimes across all datasets, demonstrating superior scalability and practical applicability for large temporal networks. 
\begin{figure}[H]
\centering
\caption{Log-scale Runtime Comparison of Centrality Measures across Real-World Networks. Efficiency Centrality (EffC) was omitted from the MathOverflow runtime curve because its execution exceeded practical computational time limits on the dataset.}

\begin{subfigure}[t]{0.32\linewidth}
    \centering
    \includegraphics[width=\linewidth]{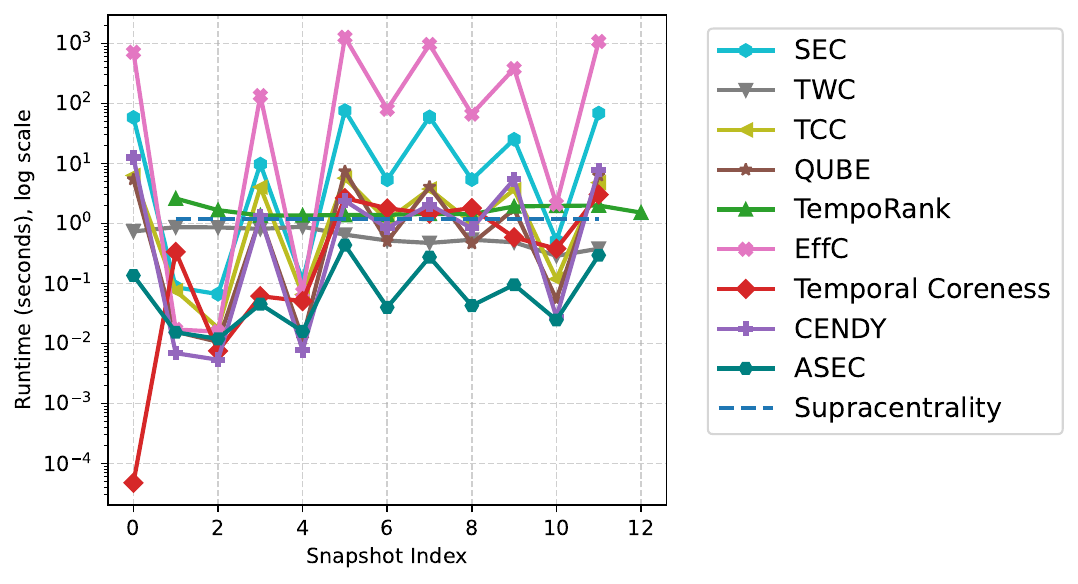}
    \caption{Email Enron Network}
    \label{fig:runtime_plot_EE}
\end{subfigure}
\hfill
\begin{subfigure}[t]{0.32\linewidth}
    \centering
    \includegraphics[width=\linewidth]{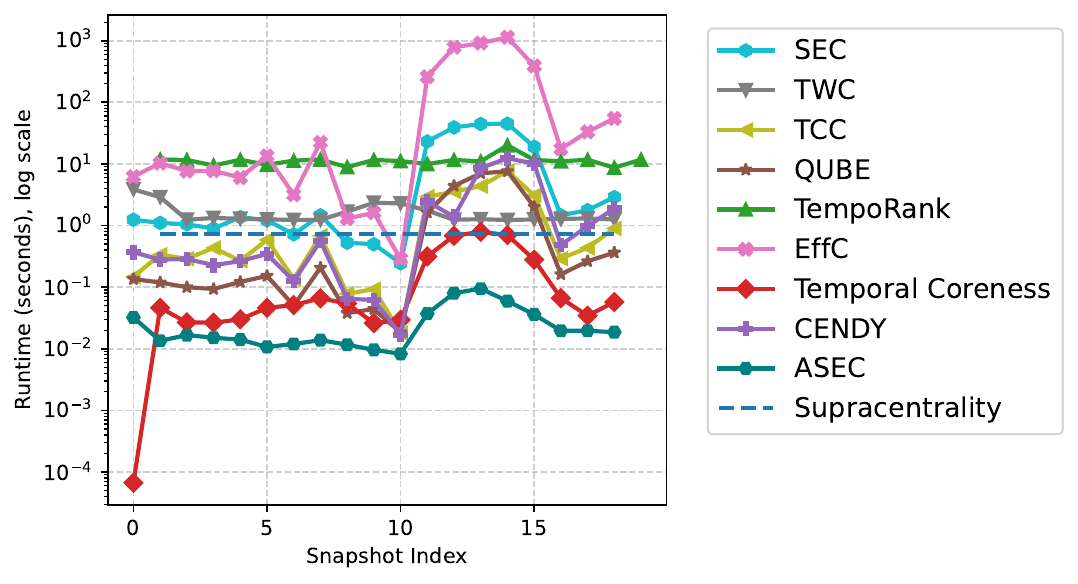}
    \caption{CollegeMsg Network}
    \label{fig:runtime_plot_CM}
\end{subfigure}
\hfill
\begin{subfigure}[t]{0.32\linewidth}
    \centering
    \includegraphics[width=\linewidth]{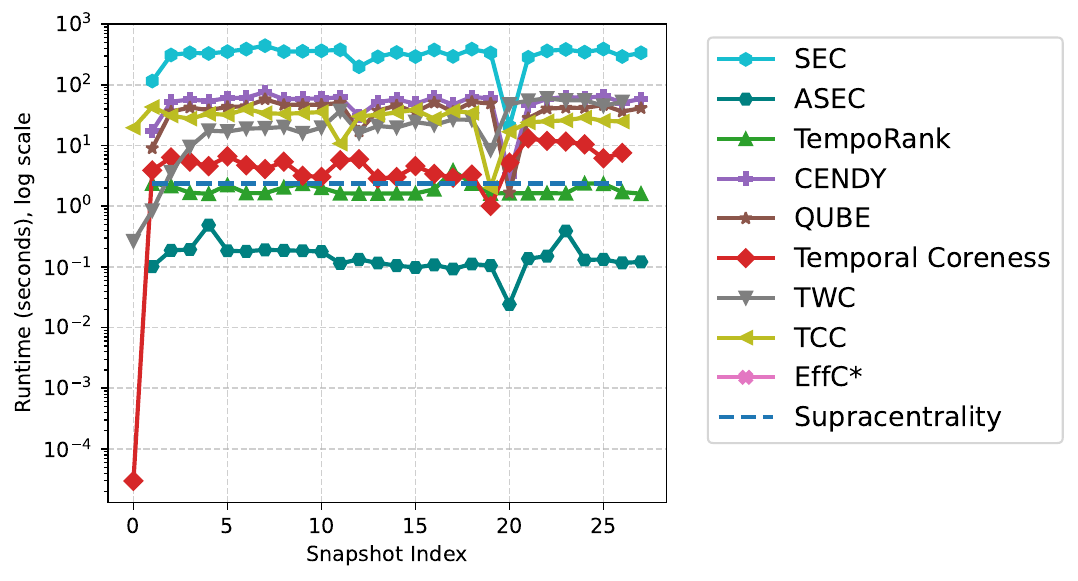}
    \caption{Math Overflow Network}
    \label{fig:runtime_plot_MO}
\end{subfigure}

\label{fig:SEC_vs_ASEC_vs_BCM_runtime}
\end{figure}

\subsection{Empirical Accuracy Analysis of ASEC}
\label{subsec:Emeracc}
To evaluate the practical validity of the proposed ASEC, we empirically compare it with SEC using all three real-world temporal datasets. We evaluate the agreement in centrality values, boundedness of absolute error, and the preservation of node ranking of both methods. These experiments demonstrate that the Perron–Frobenius–based approximation (ASEC) closely reproduces the behavior of SEC while reducing the runtime by orders of magnitude. They further validate the theoretical analysis presented in Section \ref{subsec:ASEC}. Specifically, they evaluate whether the neglected higher-order perturbation terms remain sufficiently small in practice, thereby confirming the accuracy of the first-order approximation underlying ASEC.
\subsubsection{Agreement of ASEC with exact SEC:}
ASEC was computed using a single Perron–Frobenius eigenvector per snapshot, in contrast to SEC which requires repeated spectral radius recomputation for each node. Figure \ref{fig:SEC_vs_ASEC_scatterplot} compares the relationship between $C_{SEC}$ and $C_{ASEC}$ values across different networks. Each point corresponds to a node and for each node, we explicitly compute the $C_{SEC}$ and $C_{ASEC}$ values, with the dashed diagonal indicating perfect agreement. The strong alignment of points along the diagonal demonstrates that the ASEC closely reproduces SEC values. Notably, nodes with low $C_{SEC}$ values are particularly well approximated, while small deviations occur primarily among high-impact nodes. Importantly, despite these deviations, ASEC preserves the relative importance of nodes while substantially reducing computational cost.
\begin{figure}[H]
    \centering
    \caption{SEC vs ASEC : Node-Level Agreement.}
     
    \begin{subfigure}[t]{0.32\linewidth}
    \centering
    \includegraphics[width=\linewidth]{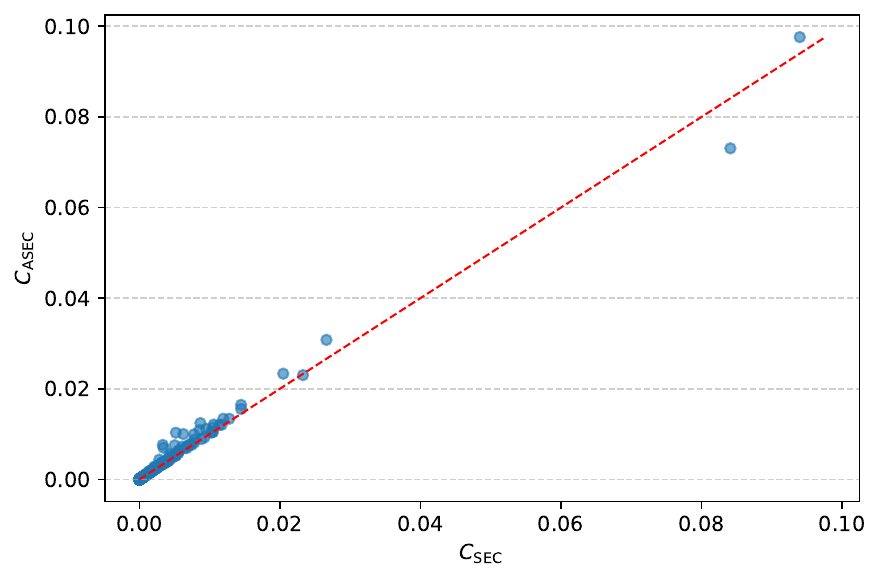}
    \caption{Email Enron Network}
    \label{fig:scatterplotEE}
    \end{subfigure}
    \hfill
    \begin{subfigure}[t]{0.32\linewidth}
    \centering
\includegraphics[width=\linewidth]{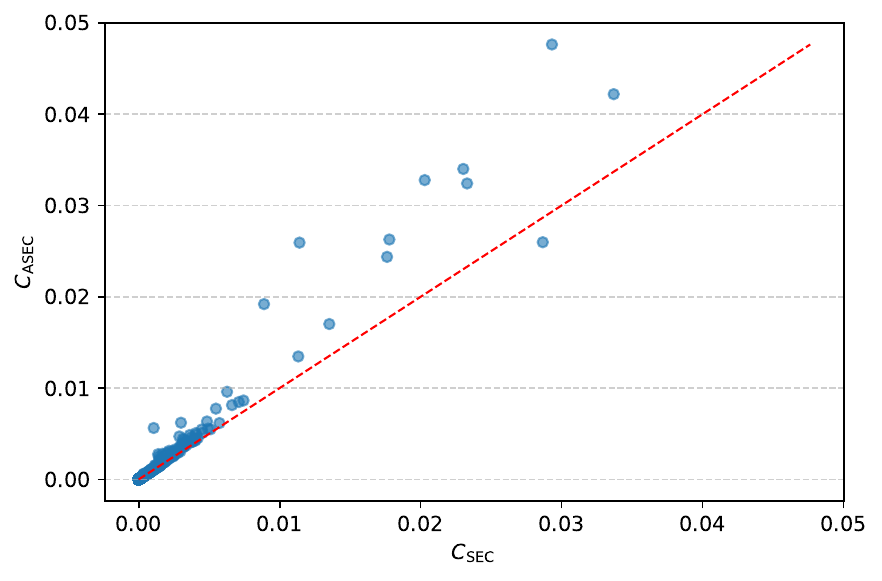}
    \caption{CollegeMsg Network}
    \label{fig:scatterplotCM}
    \end{subfigure}
    \hfill
     \begin{subfigure}[t]{0.32\linewidth}
    \centering
    \includegraphics[width=\linewidth]{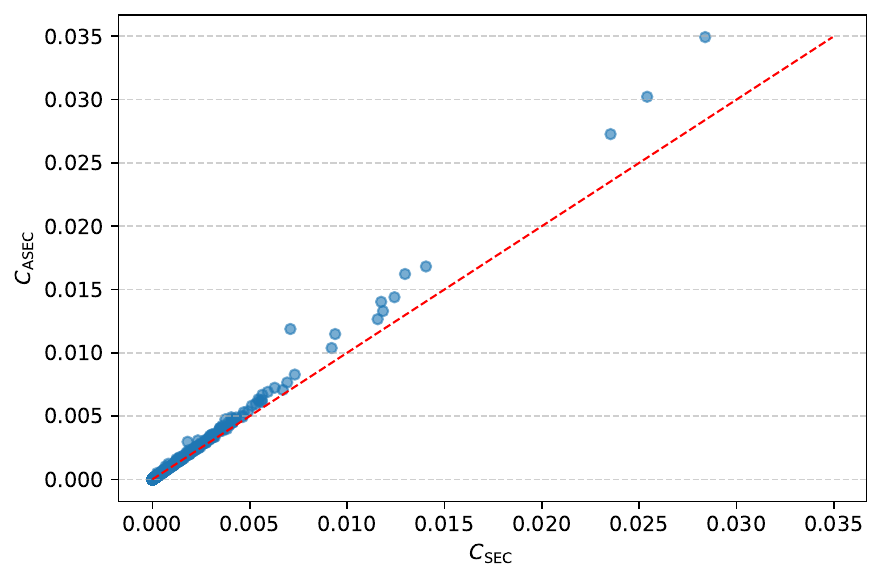}
    \caption{Math Overflow Network}
    \label{fig:scatterplotMO}
    \end{subfigure}
    \hfill

\label{fig:SEC_vs_ASEC_scatterplot}  
\end{figure}
 
\subsubsection{Boundedness of Approximation error:} 
To evaluate the correctness of the proposed approximation(ASEC), we compute its absolute deviation from SEC for every node in the network: $|\mathrm{C_{SEC}}(v) - \mathrm{C_{ASEC}}(v)|$. This error is computed using the aggregated $C_{SEC}$ and $C_{ASEC}$ values across all snapshots of the temporal network. Figure \ref{fig:SEC_vs_ASEC_abserror} reports the absolute approximation error between SEC and ASEC across three real-world networks. \par
As can be observed, the absolute error remains close to zero for the majority of nodes, with only a small number of localized deviations.. Furthermore, the maximum observed error remains below approximately $1.8 * 10^{-2}$
 across all three datasets, indicating that the approximation error is consistently small despite substantial differences in network size and structure. Since Proposition 1 shows that the approximation error is equal to the normalized higher-order remainder term ($\frac{R_v}{\lambda(A)}$), the consistently small errors observed across all three datasets indicate that the contribution of higher-order terms is negligible in practice. Although slightly larger deviations are observed for few nodes, the approximation remains accurate and preserves the overall ranking of influential nodes.
\begin{figure}[H]
    \centering
     \caption{Node-wise Absolute Approximation Error between SEC and ASEC across three Real-World datasets.}
    
     \begin{subfigure}[t]{0.32\linewidth}
    \centering
    \includegraphics[width=\linewidth]{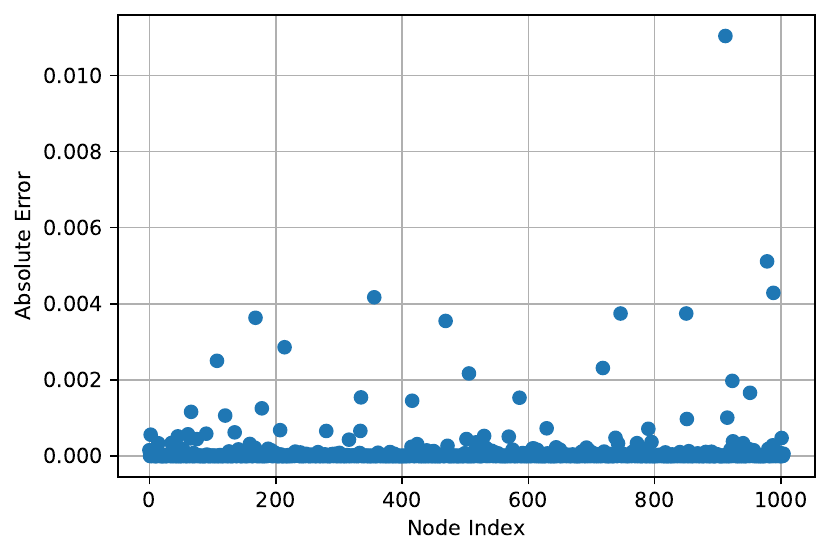}
    \caption{Email Enron Network}
    \label{fig:abserrorEE}
    \end{subfigure}
    \hfill
    \begin{subfigure}[t]{0.32\linewidth}
    \centering
    \includegraphics[width=\linewidth]{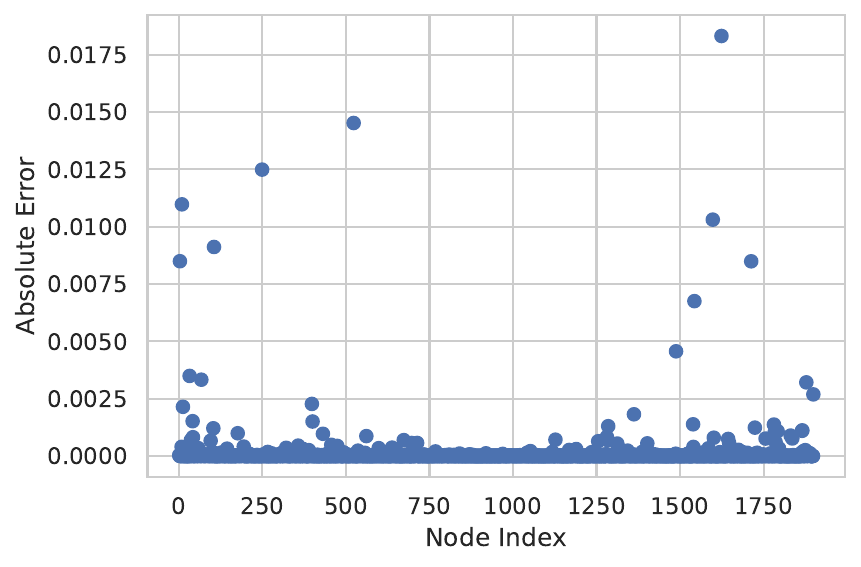}
    \caption{CollegeMsg Network}
    \label{fig:abserrorCM}
    \end{subfigure}
    \hfill
    \begin{subfigure}[t]{0.32\linewidth}
    \centering
    \includegraphics[width=\linewidth]{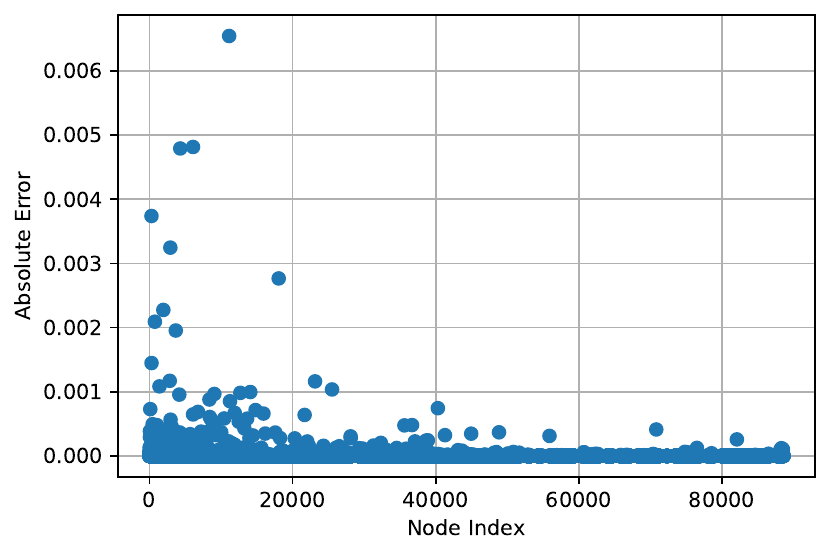}
    \caption{Math Overflow Network}
    \label{fig:abserrorMO}
    \end{subfigure}
 
\label{fig:SEC_vs_ASEC_abserror}  
\end{figure}
\subsubsection{Ranking Preservation Analysis:} Centrality measures are primarily used for identifying and ranking influential nodes, therefore it is crucial that the approximation maintains both the top influential set and the overall ordering of nodes. This subsection evaluates how well the approximate formulation (ASEC) preserves the ranking structure of the exact SEC.
\begin{enumerate}
\item {Top-k Overlap:} We further evaluate the accuracy of ASEC using a Top-$k$ overlap experiment, which measures agreement among the most influential nodes. For a given $k$, we find the number of nodes that appear in the top-$k$ rankings produced by both SEC and ASEC. Table \ref{tab:topk_overlapEE} and \ref{tab:topk_overlapMO}  depicts that ASEC consistently identifies the same highly influential nodes as SEC for $k=5$. Similarily, in Table \ref{tab:topk_overlapCM}, ASEC achieves over $90\%$ overlap for larger values of $k$ in CollegeMsg dataset. These results indicate that the approximation reliably preserves the most influential nodes while significantly reducing computational complexity.
\begin{table}[H]
\caption{Top-$k$ overlap between SEC and ASEC across different datasets.}
\centering

\begin{subtable}[t]{0.32\linewidth}
\centering
\begin{tabular}{cc}
\hline
\textbf{$k$} & \textbf{Overlap} \\
\hline
5  & 1.00 \\
10 & 0.90 \\
20 & 0.90 \\
50 & 0.96 \\
\hline
\end{tabular}
\caption{Email Enron Network}
\label{tab:topk_overlapEE}
\end{subtable}
\hfill
\begin{subtable}[t]{0.32\linewidth}
\centering
\begin{tabular}{cc}
\hline
\textbf{$k$} & \textbf{Overlap} \\
\hline
5  & 0.80 \\
10 & 0.90 \\
20 & 0.90 \\
50 & 0.96 \\
\hline
\end{tabular}
\caption{CollegeMsg Network}
\label{tab:topk_overlapCM}
\end{subtable}
\hfill
\begin{subtable}[t]{0.32\linewidth}
\centering
\begin{tabular}{cc}
\hline
\textbf{$k$} & \textbf{Overlap} \\
\hline
5  & 1.00 \\
10 & 0.90 \\
20 & 0.95 \\
50 & 0.96 \\
\hline
\end{tabular}
\caption{Math Overflow Network}
\label{tab:topk_overlapMO}
\end{subtable}

\label{tab:topk_overlap}
\end{table}
\item {Kendall $\tau$ Correlation:}
To evaluate the extent to which the ASEC preserves the global ranking consistency across all nodes, we computed Kendall’s $\tau$ rank correlation coefficient between $C_{SEC}$ and $C_{ASEC}$ scores for all datasets. The correlation values are given in Table \ref{tab:kendallcorrelation}. 
\begin{table}[H]
\centering
\caption{ Kendall $\tau$ Correlation between SEC and ASEC across various datasets}
\label{tab:kendallcorrelation}
\begin{tabular}{lc}
\hline
\textbf{Dataset} & \textbf{$\tau$ } \\
\hline

Enron Email & 0.9964 \\
CollegeMsg & 0.9931 \\
Math Overflow & 0.9950 \\
\hline
\end{tabular}
\end{table}
\end{enumerate}
These values provide complementary evidence supporting the validity of ASEC across diverse network structures. Such consistently high correlation confirms that ASEC preserves the relative ordering of nodes with negligible distortion, confirming its reliability as a computationally efficient alternative to the exact spectral method.\par

\subsection{Performance Analysis}

This section evaluates the spreading capability of SEC and its approximation ASEC in comparison with the baseline centrality measures under the SI, SIS, and IC diffusion models. To examine the consistency of the proposed methods under different diffusion settings, a parameter sensitivity analysis was performed by varying the infection probability ($\beta={0.0001,0.002,0.05,0.1,0.2,0.3}$) for the SI and SIS models, the recovery probability ($\gamma={0.02,0.04,0.1,0.2,0.3,0.4}$) for the SIS model, and the activation probability ($p={0.01,0.02,0.03,0.04,0.05,0.06}$) for the IC model. For each parameter setting, 10000 Monte Carlo simulations were performed, and the Area Under the Infection Curve (AUC) was computed to quantify the cumulative spreading efficiency. The average AUC values across all parameter settings are reported in Table \ref{tab:combined_auc}.
\begin{table}[H]
\centering
\caption{Cumulative spreading efficiency (AUC) of different centrality measures under various diffusion models for three dynamic datasets: Enron Email, CollegeMsg and MathOverflow. Higher values indicate stronger spreading capability. ""–"” denotes unavailable results for EffC on the MathOverflow dataset, as the method exceeded practical computational time limits during execution.}
\label{tab:combined_auc}
\resizebox{\textwidth}{!}{
\begin{tabular}{l|ccc|ccc|ccc}
\hline
&
\multicolumn{3}{c|}{\textbf{Enron Email}}
&
\multicolumn{3}{c|}{\textbf{CollegeMsg}}
&
\multicolumn{3}{c}{\textbf{MathOverflow}}\\

\cline{2-10}

\textbf{Centrality}
&
\textbf{SI}
&
\textbf{SIS}
&
\textbf{IC}
&
\textbf{SI}
&
\textbf{SIS}
&
\textbf{IC}
&
\textbf{SI}
&
\textbf{SIS}
&
\textbf{IC}
\\
\hline

SEC
& \textbf{0.0724} & \textbf{0.0326} & \textbf{0.0671}
&0.0991  & \textbf{0.0090} & \textbf{5.5399}
& \textbf{0.0268} & \textbf{0.0050} & \textbf{0.0285} 
\\

ASEC
& \textbf{0.0723} & \textbf{0.0325} & \textbf{0.0671}
& 0.0988 & \textbf{0.0090} & 5.4624
& \textbf{0.0268} & \textbf{0.0049}  & \textbf{0.0285} \\

TWC \cite{oettershagen2022temporal}
& 0.0704 & 0.0318 & 0.0660
& 0.1008  & 0.0087 & 5.5351
& 0.0266  & 0.0049 & 0.0284 \\

EffC \cite{wang2017new}
& 0.0700 & 0.0323 & 0.0662
& 0.0970  & 0.0087 & 5.1154
& - & - & - \\

Supracentrality  \cite{taylor2017eigenvector}
& 0.0558  & 0.0280 & 0.0558
& 0.0948 & 0.0086 & 0.0997
& 0.0157  & 0.0043 & 0.0157 \\

TCC \cite{takaguchi2016coverage}
& 0.0655 & 0.0306 & 0.0626
& 0.0960  & 0.0085 & 5.3805
& 0.0204  & 0.0046 & 0.0210 \\

TempoRank\cite{rocha2014random}
& 0.0705 & 0.0316 & 0.0661
& 0.0958  & 0.0087 & 4.9376
& 0.0221  & 0.0046 & 0.0230 \\

Temporal Coreness \cite{li2013efficient}
& 0.0692  & 0.0313 & 0.0653
& 0.0968 & 0.0088 & 5.2248
& 0.0206 & 0.0045 & 0.0215 \\

CENDY \cite{yen2013efficient} 
& 0.0711 & 0.0318 & 0.0662
& 0.0963 & 0.0088 & 4.1824
& 0.0172  & 0.0044 & 0.0173 \\

QUBE \cite{lee2012qube}
& 0.0671 & 0.0305 & 0.0637
& 0.0970  & 0.0089 & 5.3247
& 0.0245  & 0.0049 & 0.0258 \\

\hline
\end{tabular}}
\end{table}
 
Overall, the proposed SEC consistently demonstrates superior spreading performance across the Enron Email and MathOverflow datasets, achieving the highest cumulative AUC under all three diffusion models. Its approximation, ASEC, closely reproduces the performance of SEC in every experiment, confirming that the proposed first-order perturbation approximation effectively preserves the diffusion capability of the exact method while significantly reducing computational complexity. On the CollegeMsg dataset, TWC achieves a marginally higher AUC than SEC under the SI model; however, the performance difference is negligible, and SEC achieves the highest AUC under both the SIS and IC models. Considering the substantially higher computational cost of TWC, the proposed ASEC provides a more favorable balance between spreading effectiveness and computational efficiency. Furthermore, unlike Efficiency Centrality (EffC), which could not be executed on the large-scale MathOverflow dataset due to its excessive computational requirements, both SEC and ASEC successfully scale to large temporal networks while maintaining consistently high diffusion performance. Since several competing methods produce relatively close AUC values, particularly on the CollegeMsg dataset, a statistical significance analysis is presented in the following subsection to determine whether the observed performance differences are statistically meaningful.

\subsection{Statistical Significance Analysis}
To determine whether the observed differences in diffusion performance are statistically meaningful, paired t-tests were performed between ASEC and each competing centrality measure using the snapshot-wise Mean Fraction Infected (MFI) values obtained across temporal snapshots.  The null hypothesis assumes that there is no significant difference between the snapshot-wise MFI values of ASEC and each competing centrality measure. A significance level of $p$=0.05 was adopted throughout the analysis. Since ASEC closely approximates SEC while requiring substantially lower computational cost, it was selected as the reference method for the statistical comparison.
\begin{table}[H]
\centering
\caption{Paired t-test results based on snapshot-wise Mean Fraction Infected (MFI) values. Entries report $p$-values; bold values indicate statistically significant differences ($p$<0.05). ""–"” denotes unavailable results for EffC on the MathOverflow dataset, as the method exceeded practical computational time limits during execution.  }
\label{tab:ttest_results}
\resizebox{\textwidth}{!}{
\begin{tabular}{l|ccc|ccc|ccc}
\hline
&
\multicolumn{3}{c|}{\textbf{Enron Email}}
&
\multicolumn{3}{c|}{\textbf{CollegeMsg}}
&
\multicolumn{3}{c}{\textbf{MathOverflow}}\\

\cline{2-10}

\textbf{Comparison}
&
\textbf{SI}
&
\textbf{SIS}
&
\textbf{IC}
&
\textbf{SI}
&
\textbf{SIS}
&
\textbf{IC}
&
\textbf{SI}
&
\textbf{SIS}
&
\textbf{IC}
\\
\hline

ASEC vs SEC
&
0.5675
&
0.7858
&
0.4986
&
\textbf{0.0244}
&
0.4433
&
\textbf{<0.0001}
&
0.9208
&
\textbf{<0.0001}
&
0.9538
\\

ASEC vs TWC
&
\textbf{0.0201}
&
\textbf{0.0422}
&
\textbf{0.0314}
&
\textbf{0.0084}
&
\textbf{0.0099}
&
\textbf{0.0002}
&
\textbf{0.0361}
&
\textbf{0.0492}
&
\textbf{0.0326}
\\

ASEC vs EffC
&
\textbf{0.0074}
&
\textbf{0.0238}
&
\textbf{0.0242}
&
\textbf{<0.0001}
&
\textbf{0.0026}
&
\textbf{<0.0001}
&
--
&
--
&
--
\\

ASEC vs Supra-centrality
&
\textbf{0.0010}
&
\textbf{0.0005}
&
\textbf{0.0009}
&
\textbf{<0.0001}
&
\textbf{0.0042}
&
\textbf{<0.0001}
&
\textbf{<0.0001}
&
\textbf{<0.0001}
&
\textbf{<0.0001}
\\

ASEC vs TCC
&
\textbf{0.0011}
&
\textbf{0.0026}
&
\textbf{0.0008}
&
\textbf{<0.0001}
&
\textbf{0.0012}
&
\textbf{0.0004}
&
\textbf{<0.0001}
&
\textbf{<0.0001}
&
\textbf{<0.0001}
\\

ASEC vs TempoRank
&
\textbf{0.0016}
&
\textbf{0.0140}
&
\textbf{0.0043}
&
\textbf{<0.0001}
&
\textbf{0.0016}
&
\textbf{<0.0001}
&
\textbf{<0.0001}
&
\textbf{<0.0001}
&
\textbf{<0.0001}
\\

ASEC vs Temporal Coreness
&
\textbf{0.0032}
&
\textbf{0.0024}
&
\textbf{0.0037}
&
\textbf{<0.0001}
&
\textbf{0.0062}
&
\textbf{<0.0001}
&
\textbf{<0.0001}
&
\textbf{<0.0001}
&
\textbf{<0.0001}
\\

ASEC vs CENDY
&
\textbf{0.0227}
&
\textbf{0.0470}
&
\textbf{0.0201}
&
\textbf{<0.0001}
&
\textbf{0.0204}
&
\textbf{<0.0001}
&
\textbf{<0.0001}
&
\textbf{<0.0001}
&
\textbf{<0.0001}
\\

ASEC vs QUBE
&
\textbf{0.0017}
&
\textbf{0.0020}
&
\textbf{0.0018}
&
\textbf{<0.0001}
&
\textbf{0.0171}
&
\textbf{<0.0001}
&
\textbf{<0.0001}
&
\textbf{<0.0001}
&
\textbf{<0.0001}
\\

\hline
\end{tabular}}
\end{table}
The statistical results across the Enron Email, CollegeMsg, and MathOverflow datasets strongly support the diffusion analysis presented in the previous subsection. In the vast majority of comparisons, ASEC exhibits statistically significant improvements over the competing centrality measures (p<0.05), confirming that the observed differences in diffusion performance are unlikely to be due to random variation. Furthermore, although the comparison between SEC and ASEC is statistically significant in a small number of cases, the corresponding differences in cumulative spreading efficiency remain negligible, indicating that ASEC closely reproduces the diffusion behaviour of SEC while offering substantially lower computational cost. Overall, these results provide strong statistical evidence that the proposed spectral efficiency framework consistently achieves superior or highly competitive diffusion performance across diverse temporal networks.
\subsection{Robustness Analysis} To further evaluate the structural effectiveness of the proposed centrality measures, we performed a robustness analysis based on targeted node removal. Unlike the diffusion experiments, which assess spreading capability, this experiment evaluates the ability of each centrality measure to identify structurally critical nodes whose removal most effectively disrupts the network. For each centrality measure, the top-10 nodes obtained from the global temporal ranking were identified. These nodes were then sequentially removed from every temporal snapshot, and the normalized spectral radius and normalized Largest Connected Component (LCC) of the remaining graph were computed. The robustness curves in Figure \ref{fig:Spectralradius} and \ref{fig:LargestCC} correspond to the average of these structural measures over all temporal snapshots.
\begin{figure}[H]
    \centering
     \caption{Average normalized spectral radius under sequential node removal (lower values indicate greater structural degradation). Efficiency Centrality (EffC) was omitted from the MathOverflow Network curve because its execution exceeded practical computational time limits on the dataset.}
    
     \begin{subfigure}[t]{0.32\linewidth}
    \centering
    \includegraphics[width=0.95\linewidth]{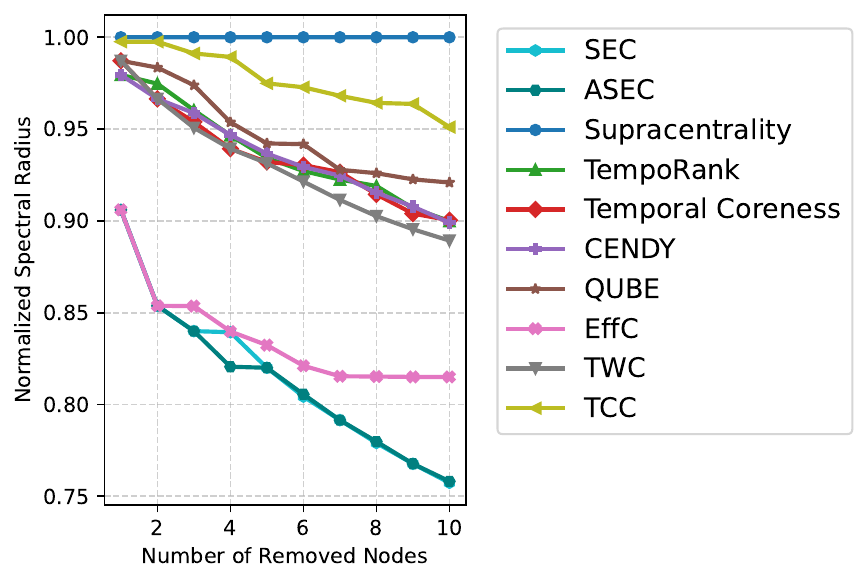}
    \caption{Email Enron Network}
    \label{fig:tempSREE}
    \end{subfigure}
    \hfill
    \begin{subfigure}[t]{0.32\linewidth}
    \centering
    \includegraphics[width=0.95\linewidth]{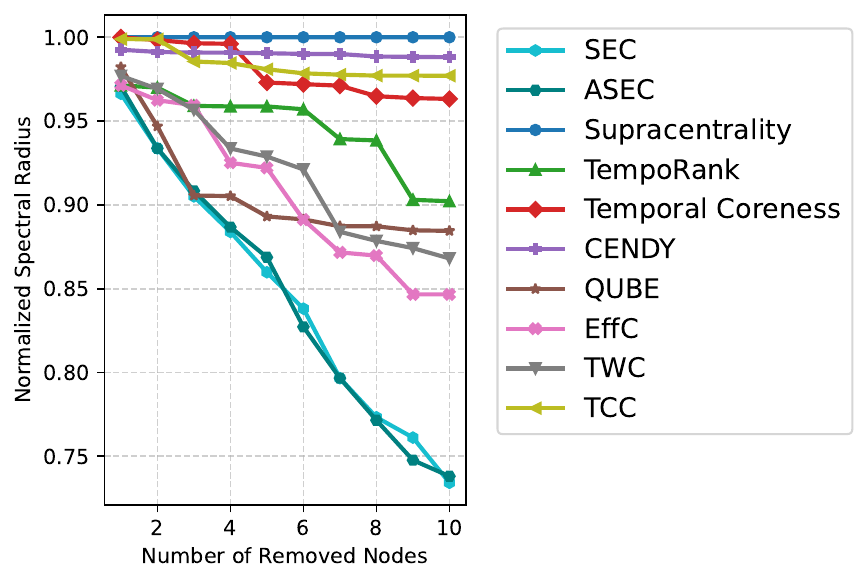}
    \caption{CollegeMsg Network}
    \label{fig:tempSRCM}
    \end{subfigure}
    \hfill
    \begin{subfigure}[t]{0.32\linewidth}
    \centering
    \includegraphics[width=0.95\linewidth]{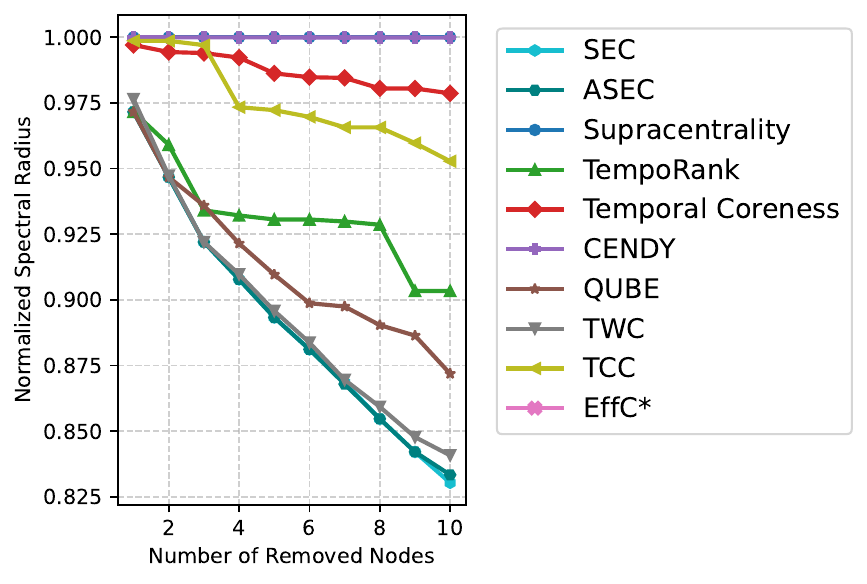}
    \caption{Math Overflow Network}
    \label{fig:tempSRMO}
    \end{subfigure}
 
\label{fig:Spectralradius}  
\end{figure}
\begin{figure}[H]
    \centering
     \caption{Average normalized largest connected component (LCC) under sequential node removal (lower values indicate greater network fragmentation). Efficiency Centrality (EffC) was omitted from the MathOverflow Network curve because its execution exceeded practical computational time limits on the dataset.}
    
     \begin{subfigure}[t]{0.32\linewidth}
    \centering
    \includegraphics[width=0.95\linewidth]{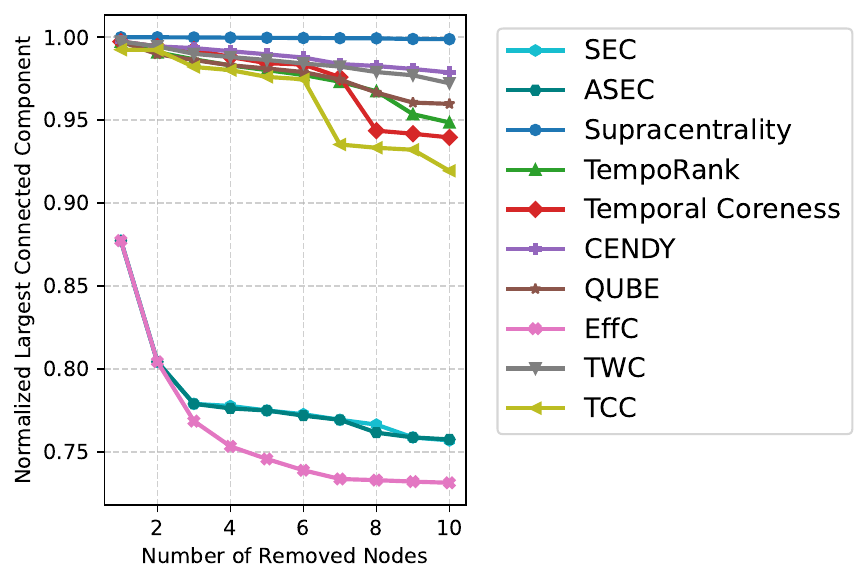}
    \caption{Email Enron Network}
    \label{fig:tempLCCEE}
    \end{subfigure}
    \hfill
    \begin{subfigure}[t]{0.32\linewidth}
    \centering
    \includegraphics[width=0.95\linewidth]{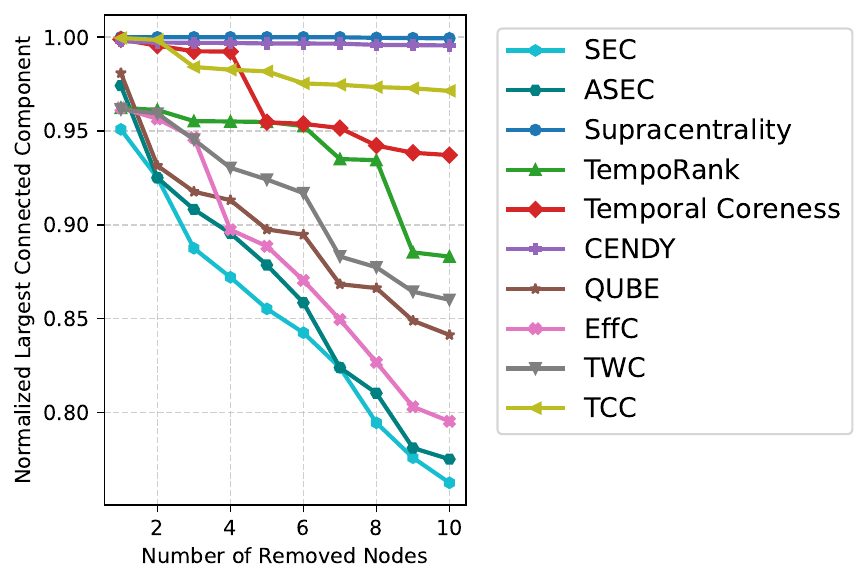}
    \caption{CollegeMsg Network}
    \label{fig:tempLCCCM}
    \end{subfigure}
    \hfill
    \begin{subfigure}[t]{0.32\linewidth}
    \centering
    \includegraphics[width=0.95\linewidth]{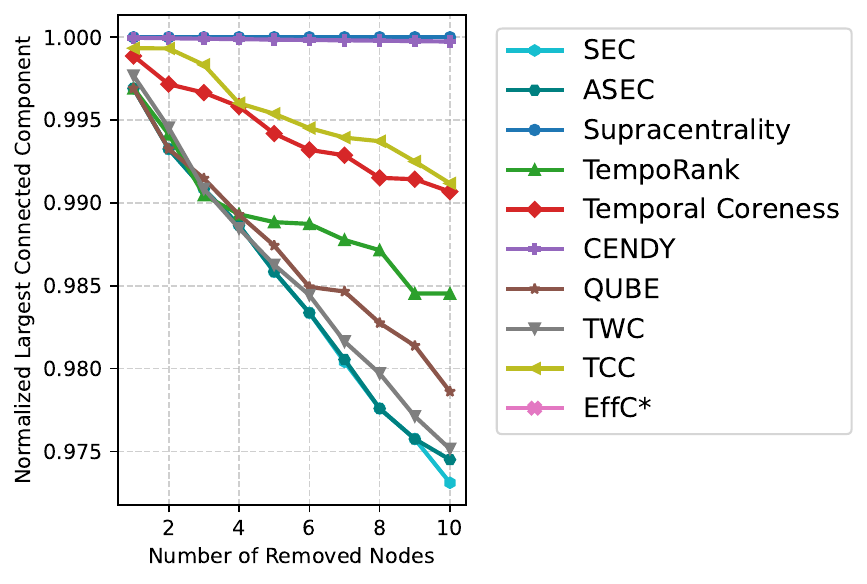}
    \caption{Math Overflow Network}
    \label{fig:tempLCCMO}
    \end{subfigure}
 
\label{fig:LargestCC}  
\end{figure}
Across all three datasets, SEC consistently exhibits the fastest decline in both structural metrics, indicating that it identifies the most structurally critical nodes whose removal causes the greatest degradation of the temporal network. ASEC closely follows the performance of SEC throughout the node removal process, demonstrating that the proposed approximation effectively preserves the robustness characteristics of the exact spectral method while substantially reducing computational cost. Although TWC achieves competitive diffusion performance, particularly on the CollegeMsg dataset, it consistently results in slower structural degradation than SEC and ASEC, suggesting that the proposed methods identify nodes that are not only effective spreaders but also more essential for maintaining network connectivity. The remaining baseline methods, including TempoRank, Temporal Coreness, CENDY, QUBE, TCC, and Supracentrality, produce comparatively smaller reductions in spectral radius and LCC, reflecting their relatively weaker ability to identify structurally influential nodes. Moreover, unlike EffC, which could not be evaluated on the large-scale MathOverflow dataset because of its high computational cost, both SEC and ASEC remain computationally feasible while maintaining superior robustness performance. Overall, these results complement the diffusion experiments by demonstrating that the proposed spectral efficiency framework captures globally influential nodes from both diffusion and structural perspectives, thereby providing a more comprehensive assessment of node importance in temporal networks.

\section{Limitations and Future Work}
SEC and ASEC demonstrate promising results in identifying structurally critical nodes in temporal networks; however, several limitations remain that also motivate directions for future
research. Addressing these aspects offers significant potential to further refine the method and expand its applicability in both theoretical and practical settings.

\begin{enumerate}
 \item{Weighted and Directed Networks:} The proposed SEC framework is currently limited to undirected and unweighted temporal networks, as the theoretical formulation relies on symmetric adjacency matrices for spectral analysis and the derivation of ASEC. Although this assumption is appropriate for many contact and communication networks, interactions are directional or vary in strength across many real-world systems. Extending SEC to weighted and directed temporal networks is a natural next step and may provide richer structural insights in such settings.

\item{Online and Streaming Scenarios:}  
SEC is currently computed in an offline manner over fixed temporal snapshots. Developing incremental or streaming variants that update $C_{SEC}$ scores as new interactions arrive would enhance its applicability to continuously evolving networks and real-time analysis tasks.

\item{Enhancing Generalizability:} The current evaluation focuses on a limited set of temporal communication data. Broader validation across diverse domains, including financial, biological, and infrastructure networks, will further strengthen the generalizability and applicability of SEC.
\end{enumerate}
\section{Conclusion}

In this work, we proposed Spectral Efficiency Centrality (SEC), a temporal centrality measure that quantifies node importance through its impact on the network spectral radius. Unlike existing temporal centrality measures, SEC captures the global structural influence of nodes and provides a spectral perspective for identifying influential spreaders in temporal networks. To improve computational efficiency, we further developed ASEC, a first-order approximation that preserves the ranking quality of SEC while substantially reducing computational cost.\par

Extensive experiments on the Enron Email, CollegeMsg, and MathOverflow datasets demonstrate that SEC and ASEC consistently achieve competitive or superior diffusion performance while remaining computationally efficient. Statistical significance and robustness analyses further confirm that the proposed methods identify nodes that are not only effective spreaders but also structurally critical to temporal networks. Overall, the proposed framework offers a theoretically grounded, scalable, and effective approach for influential node identification in temporal networks.

\bibliographystyle{plain} 
\bibliography{ref}

\begin{thebibliography}{10}

\bibitem{albert2002statistical}
R{\'e}ka Albert and Albert-L{\'a}szl{\'o} Barab{\'a}si.
\newblock Statistical mechanics of complex networks.
\newblock {\em Reviews of modern physics}, 74(1):47, 2002.

\bibitem{anderson1991infectious}
Roy~M Anderson and Robert~M May.
\newblock {\em Infectious diseases of humans: dynamics and control}.
\newblock Oxford university press, 1991.

\bibitem{boccaletti2006complex}
Stefano Boccaletti, Vito Latora, Yamir Moreno, Martin Chavez, and D-U Hwang.
\newblock Complex networks: Structure and dynamics.
\newblock {\em Physics reports}, 424(4-5):175--308, 2006.

\bibitem{bonacich1972factoring}
Phillip Bonacich.
\newblock Factoring and weighting approaches to status scores and clique identification.
\newblock {\em Journal of mathematical sociology}, 2(1):113--120, 1972.

\bibitem{bonacich1987power}
Phillip Bonacich.
\newblock Power and centrality: A family of measures.
\newblock {\em American journal of sociology}, 92(5):1170--1182, 1987.

\bibitem{brede2012networks}
Markus Brede.
\newblock Networks—an introduction. mark ej newman.(2010, oxford university press.) \$65.38,{\pounds} 35.96 (hardcover), 772 pages. isbn-978-0-19-920665-0., 2012.

\bibitem{brodka2011degree}
Piotr Br{\'o}dka, Krzysztof Skibicki, Przemys{\l}aw Kazienko, and Katarzyna Musia{\l}.
\newblock A degree centrality in multi-layered social network.
\newblock In {\em 2011 international conference on computational aspects of social networks (CASoN)}, pages 237--242. IEEE, 2011.

\bibitem{buss2020algorithmic}
Sebastian Bu{\ss}, Hendrik Molter, Rolf Niedermeier, and Maciej Rymar.
\newblock Algorithmic aspects of temporal betweenness.
\newblock In {\em Proceedings of the 26th ACM SIGKDD international conference on knowledge discovery \& data mining}, pages 2084--2092, 2020.

\bibitem{callaghan2007random}
Thomas Callaghan, Peter~J Mucha, and Mason~A Porter.
\newblock Random walker ranking for ncaa division ia football.
\newblock {\em The American Mathematical Monthly}, 114(9):761--777, 2007.

\bibitem{chartier2011sensitivity}
Timothy~P Chartier, Erich Kreutzer, Amy~N Langville, and Kathryn~E Pedings.
\newblock Sensitivity and stability of ranking vectors.
\newblock {\em SIAM Journal on Scientific Computing}, 33(3):1077--1102, 2011.

\bibitem{crescenzi2020finding}
Pierluigi Crescenzi, Cl{\'e}mence Magnien, and Andrea Marino.
\newblock Finding top-k nodes for temporal closeness in large temporal graphs.
\newblock {\em Algorithms}, 13(9):211, 2020.

\bibitem{cruciani2024mantra}
Antonio Cruciani.
\newblock Mantra: Temporal betweenness centrality approximation through sampling.
\newblock In {\em Joint European Conference on Machine Learning and Knowledge Discovery in Databases}, pages 125--143. Springer, 2024.

\bibitem{daly2008social}
Elizabeth~M Daly and Mads Haahr.
\newblock Social network analysis for information flow in disconnected delay-tolerant manets.
\newblock {\em IEEE Transactions on Mobile Computing}, 8(5):606--621, 2008.

\bibitem{estrada2013communicability}
Ernesto Estrada.
\newblock Communicability in temporal networks.
\newblock {\em Physical Review E—Statistical, Nonlinear, and Soft Matter Physics}, 88(4):042811, 2013.

\bibitem{fortunato2010community}
Santo Fortunato.
\newblock Community detection in graphs.
\newblock {\em Physics reports}, 486(3-5):75--174, 2010.

\bibitem{freeman1978centrality}
Linton~C Freeman.
\newblock Centrality in social networks conceptual clarification.
\newblock {\em Social networks}, 1(3):215--239, 1978.

\bibitem{ghosh2011parameterized}
Rumi Ghosh and Kristina Lerman.
\newblock Parameterized centrality metric for network analysis.
\newblock {\em Physical Review E—Statistical, Nonlinear, and Soft Matter Physics}, 83(6):066118, 2011.

\bibitem{girvan2002community}
Michelle Girvan and Mark~EJ Newman.
\newblock Community structure in social and biological networks.
\newblock {\em Proceedings of the national academy of sciences}, 99(12):7821--7826, 2002.

\bibitem{golub2013matrix}
Gene~H Golub and Charles~F Van~Loan.
\newblock {\em Matrix computations}.
\newblock JHU press, 2013.

\bibitem{grindrod2012models}
Peter Grindrod and Desmond~J Higham.
\newblock Models for evolving networks: with applications in telecommunication and online activities.
\newblock {\em IMA Journal of Management Mathematics}, 23(1):1--15, 2012.

\bibitem{guimera2005worldwide}
Roger Guimera, Stefano Mossa, Adrian Turtschi, and LA~Nunes Amaral.
\newblock The worldwide air transportation network: Anomalous centrality, community structure, and cities' global roles.
\newblock {\em Proceedings of the National Academy of Sciences}, 102(22):7794--7799, 2005.

\bibitem{holme2003congestion}
Petter Holme.
\newblock Congestion and centrality in traffic flow on complex networks.
\newblock {\em Advances in Complex Systems}, 6(02):163--176, 2003.

\bibitem{holme2012temporal}
Petter Holme and Jari Saram{\"a}ki.
\newblock Temporal networks.
\newblock {\em Physics reports}, 519(3):97--125, 2012.

\bibitem{hu2010measuring}
Yanqing Hu, Yuchao Nie, Hua Yang, Jie Cheng, Ying Fan, and Zengru Di.
\newblock Measuring the significance of community structure in complex networks.
\newblock {\em Physical Review E—Statistical, Nonlinear, and Soft Matter Physics}, 82(6):066106, 2010.

\bibitem{kato1966perturbation}
Tosio Kato and Tosio Kat{\aa}o.
\newblock {\em Perturbation theory for linear operators}, volume 132.
\newblock Springer, 1966.

\bibitem{katz1953new}
Leo Katz.
\newblock A new status index derived from sociometric analysis.
\newblock {\em Psychometrika}, 18(1):39--43, 1953.

\bibitem{kempe2003maximizing}
David Kempe, Jon Kleinberg, and {\'E}va Tardos.
\newblock Maximizing the spread of influence through a social network.
\newblock In {\em Proceedings of the ninth ACM SIGKDD international conference on Knowledge discovery and data mining}, pages 137--146, 2003.

\bibitem{kim2012temporal}
Hyoungshick Kim and Ross Anderson.
\newblock Temporal node centrality in complex networks.
\newblock {\em Physical Review E—Statistical, Nonlinear, and Soft Matter Physics}, 85(2):026107, 2012.

\bibitem{kitsak2010identification}
Maksim Kitsak, Lazaros~K Gallos, Shlomo Havlin, Fredrik Liljeros, Lev Muchnik, H~Eugene Stanley, and Hern{\'a}n~A Makse.
\newblock Identification of influential spreaders in complex networks.
\newblock {\em Nature physics}, 6(11):888--893, 2010.

\bibitem{kostakos2009temporal}
Vassilis Kostakos.
\newblock Temporal graphs.
\newblock {\em Physica A: Statistical Mechanics and its Applications}, 388(6):1007--1023, 2009.

\bibitem{lancichinetti2010characterizing}
Andrea Lancichinetti, Mikko Kivel{\"a}, Jari Saram{\"a}ki, and Santo Fortunato.
\newblock Characterizing the community structure of complex networks.
\newblock {\em PloS one}, 5(8):e11976, 2010.

\bibitem{lanczos1950iteration}
Cornelius Lanczos.
\newblock An iteration method for the solution of the eigenvalue problem of linear differential and integral operators.
\newblock {\em Journal of research of the National Bureau of Standards}, 45(4):255--282, 1950.

\bibitem{latora2001efficient}
Vito Latora and Massimo Marchiori.
\newblock Efficient behavior of small-world networks.
\newblock {\em Physical review letters}, 87(19):198701, 2001.

\bibitem{lee2012qube}
Min-Joong Lee, Jungmin Lee, Jaimie~Yejean Park, Ryan~Hyun Choi, and Chin-Wan Chung.
\newblock {QUBE}: a quick algorithm for updating betweenness centrality.
\newblock In {\em Proceedings of the 21st international conference on World Wide Web}, pages 351--360, 2012.

\bibitem{lerman2010centrality}
Kristina Lerman, Rumi Ghosh, and Jeon~Hyung Kang.
\newblock Centrality metric for dynamic networks.
\newblock In {\em Proceedings of the Eighth Workshop on Mining and Learning with Graphs}, pages 70--77, 2010.

\bibitem{li2013efficient}
Rong-Hua Li, Jeffrey~Xu Yu, and Rui Mao.
\newblock Efficient core maintenance in large dynamic graphs.
\newblock {\em IEEE transactions on knowledge and data engineering}, 26(10):2453--2465, 2013.

\bibitem{lloyd2009epidemic}
James~O Lloyd-Smith, Dylan George, Kim~M Pepin, Virginia~E Pitzer, Juliet~RC Pulliam, Andrew~P Dobson, Peter~J Hudson, and Bryan~T Grenfell.
\newblock Epidemic dynamics at the human-animal interface.
\newblock {\em science}, 326(5958):1362--1367, 2009.

\bibitem{lu2016vital}
Linyuan L{\"u}, Duanbing Chen, Xiao-Long Ren, Qian-Ming Zhang, Yi-Cheng Zhang, and Tao Zhou.
\newblock Vital nodes identification in complex networks.
\newblock {\em Physics reports}, 650:1--63, 2016.

\bibitem{meyer2023matrix}
Carl~D Meyer.
\newblock {\em Matrix analysis and applied linear algebra}.
\newblock SIAM, 2023.

\bibitem{milanese2010approximating}
Attilio Milanese, Jie Sun, and Takashi Nishikawa.
\newblock Approximating spectral impact of structural perturbations in large networks.
\newblock {\em Physical Review E—Statistical, Nonlinear, and Soft Matter Physics}, 81(4):046112, 2010.

\bibitem{naima2025temporal}
Mehdi Naima.
\newblock Temporal betweenness centrality on shortest walks variants.
\newblock {\em Applied Network Science}, 10(1):11, 2025.

\bibitem{newman2003structure}
Mark~EJ Newman.
\newblock The structure and function of complex networks.
\newblock {\em SIAM review}, 45(2):167--256, 2003.

\bibitem{nicosia2013graph}
Vincenzo Nicosia, John Tang, Cecilia Mascolo, Mirco Musolesi, Giovanni Russo, and Vito Latora.
\newblock Graph metrics for temporal networks.
\newblock In {\em Temporal networks}, pages 15--40. Springer, 2013.

\bibitem{oettershagen2022computing}
Lutz Oettershagen and Petra Mutzel.
\newblock Computing top-k temporal closeness in temporal networks.
\newblock {\em Knowledge and Information Systems}, 64(2):507--535, 2022.

\bibitem{oettershagen2022temporal}
Lutz Oettershagen, Petra Mutzel, and Nils~M Kriege.
\newblock Temporal walk centrality: ranking nodes in evolving networks.
\newblock In {\em Proceedings of the ACM Web conference 2022}, pages 1640--1650, 2022.

\bibitem{opsahl2010node}
Tore Opsahl, Filip Agneessens, and John Skvoretz.
\newblock Node centrality in weighted networks: Generalizing degree and shortest paths.
\newblock {\em Social networks}, 32(3):245--251, 2010.

\bibitem{pan2011path}
Raj~Kumar Pan and Jari Saram{\"a}ki.
\newblock Path lengths, correlations, and centrality in temporal networks.
\newblock {\em Physical Review E—Statistical, Nonlinear, and Soft Matter Physics}, 84(1):016105, 2011.

\bibitem{perron2007perron}
Oskar Perron.
\newblock The perron-frobenius theorem.
\newblock {\em Cit. on}, page~27, 2007.

\bibitem{restrepo2006characterizing}
Juan~G Restrepo, Edward Ott, and Brian~R Hunt.
\newblock Characterizing the dynamical importance of network nodes and links.
\newblock {\em Physical review letters}, 97(9):094102, 2006.

\bibitem{rocha2014random}
Luis~EC Rocha and Naoki Masuda.
\newblock Random walk centrality for temporal networks.
\newblock {\em New Journal of Physics}, 16(6):063023, 2014.

\bibitem{saavedra2010mutually}
Serguei Saavedra, Scott Powers, Trent McCotter, Mason~A Porter, and Peter~J Mucha.
\newblock Mutually-antagonistic interactions in baseball networks.
\newblock {\em Physica A: Statistical Mechanics and its Applications}, 389(5):1131--1141, 2010.

\bibitem{sabidussi1966centrality}
Gert Sabidussi.
\newblock The centrality index of a graph.
\newblock {\em Psychometrika}, 31(4):581--603, 1966.

\bibitem{sariyuce2013streaming}
Ahmet~Erdem Sar{\'\i}y{\"u}ce, Bu{\u{g}}ra Gedik, Gabriela Jacques-Silva, Kun-Lung Wu, and {\"U}mit~V {\c{C}}ataly{\"u}rek.
\newblock Streaming algorithms for k-core decomposition.
\newblock {\em Proceedings of the VLDB Endowment}, 6(6):433--444, 2013.

\bibitem{stewart1990matrix}
Gilbert~W Stewart and Ji-guang Sun.
\newblock Matrix perturbation theory.
\newblock {\em (No Title)}, 1990.

\bibitem{sun2007comparative}
Lanfang Sun, Menghui Li, Lu~Jiang, and Lu~Tan.
\newblock Comparative analysis of the gene co-regulatory network of normal and cancerous lung.
\newblock {\em Physica A: Statistical Mechanics and its Applications}, 384(2):739--746, 2007.

\bibitem{takaguchi2016coverage}
Taro Takaguchi, Yosuke Yano, and Yuichi Yoshida.
\newblock Coverage centralities for temporal networks.
\newblock {\em The European Physical Journal B}, 89(2):35, 2016.

\bibitem{tang2010analysing}
John Tang, Mirco Musolesi, Cecilia Mascolo, Vito Latora, and Vincenzo Nicosia.
\newblock Analysing information flows and key mediators through temporal centrality metrics.
\newblock In {\em Proceedings of the 3rd workshop on social network systems}, pages 1--6, 2010.

\bibitem{taylor2017eigenvector}
Dane Taylor, Sean~A Myers, Aaron Clauset, Mason~A Porter, and Peter~J Mucha.
\newblock Eigenvector-based centrality measures for temporal networks.
\newblock {\em Multiscale Modeling \& Simulation}, 15(1):537--574, 2017.

\bibitem{taylor2021tunable}
Dane Taylor, Mason~A Porter, and Peter~J Mucha.
\newblock Tunable eigenvector-based centralities for multiplex and temporal networks.
\newblock {\em Multiscale Modeling \& Simulation}, 19(1):113--147, 2021.

\bibitem{uddin2011time}
Shahadat Uddin and Liaquat Hossain.
\newblock Time scale degree centrality: A time-variant approach to degree centrality measures.
\newblock In {\em 2011 International Conference on Advances in Social Networks Analysis and Mining}, pages 520--524. IEEE, 2011.

\bibitem{vukadinovic2013centrality}
Danica Vukadinovi{\'c}~Greetham, Zhivko Stoyanov, and Peter Grindrod.
\newblock Centrality and spectral radius in dynamic communication networks.
\newblock In {\em International Computing and Combinatorics Conference}, pages 791--800. Springer, 2013.

\bibitem{wang2017new}
Shasha Wang, Yuxian Du, and Yong Deng.
\newblock A new measure of identifying influential nodes: Efficiency centrality.
\newblock {\em Communications in Nonlinear Science and Numerical Simulation}, 47:151--163, 2017.

\bibitem{wang2011identifying}
Yang Wang, Zengru Di, and Ying Fan.
\newblock Identifying and characterizing nodes important to community structure using the spectrum of the graph.
\newblock {\em PloS one}, 6(11):e27418, 2011.

\bibitem{wasserman1994social}
Stanley Wasserman and Katherine Faust.
\newblock Social network analysis: Methods and applications.
\newblock 1994.

\bibitem{yen2013efficient}
Chia-Chen Yen, Mi-Yen Yeh, and Ming-Syan Chen.
\newblock An efficient approach to updating closeness centrality and average path length in dynamic networks.
\newblock In {\em 2013 IEEE 13th International Conference on Data Mining}, pages 867--876. IEEE, 2013.

\bibitem{yin2018inter}
Ran-Ran Yin, Qiang Guo, Jian-Nan Yang, and Jian-Guo Liu.
\newblock Inter-layer similarity-based eigenvector centrality measures for temporal networks.
\newblock {\em Physica A: Statistical Mechanics and its Applications}, 512:165--173, 2018.

\bibitem{ZAHOOR2026115648}
Aaqib Zahoor, Janibul Bashir, and Iqra~Altaf Gillani.
\newblock Forecast to influence: Scalable temporal influence maximization in streaming settings.
\newblock {\em Engineering Applications of Artificial Intelligence}, 181:115648, 2026.

\bibitem{zhao2023general}
Xiuming Zhao, Hongtao Yu, Shuxin Liu, and Xiaochun Cao.
\newblock A general higher-order supracentrality framework based on motifs of temporal networks and multiplex networks.
\newblock {\em Physica A: Statistical Mechanics and its Applications}, 614:128548, 2023.

\end{thebibliography}

\end{document}